\documentclass[11pt]{article}
    
    \usepackage[margin=1.0in]{geometry}
    \usepackage{mathtools}
    \usepackage{amsthm}
    \usepackage{hyperref}
    \usepackage{subcaption}
    \usepackage{amsfonts}
    \usepackage{graphicx}
    \usepackage{cmbright}
    \usepackage{amssymb}
    \usepackage{fancyhdr}
    \usepackage{cleveref}
    
    \allowdisplaybreaks

    \newtheorem{theorem}{Theorem}[section]
    \newtheorem{corollary}{Corollary}[theorem]
    \newtheorem{lemma}[theorem]{Lemma}
    \newtheorem{defn}[theorem]{Definition}
    \newtheorem{prop}[theorem]{Proposition}
    \newtheorem{remark}[theorem]{Remark}

    \newcommand{\R}{\mathbb{R}}
    
    \newcommand{\diff}{\textnormal{d}}

    \newcommand{\bfa}{\boldsymbol{a}}
    \newcommand{\bfb}{\boldsymbol{b}}
    \newcommand{\bff}{\boldsymbol{f}}
    \newcommand{\bfk}{\boldsymbol{k}}
    \newcommand{\bfu}{\boldsymbol{u}}
    \newcommand{\bfv}{\boldsymbol{v}}
    \newcommand{\bfx}{\boldsymbol{x}}
    \newcommand{\bfy}{\boldsymbol{y}}
    \newcommand{\bfz}{\boldsymbol{z}}

    \newcommand{\bfA}{\boldsymbol{A}}
    \newcommand{\bfE}{\boldsymbol{E}}
    \newcommand{\bfT}{\boldsymbol{T}}
    \newcommand{\bfU}{\boldsymbol{U}}
    \newcommand{\bfW}{\boldsymbol{W}}
    
    \newcommand{\bfxi}{\boldsymbol{\xi}}

    \newcommand{\Tmax}{T_{\textnormal{max}}}
    \newcommand{\Tmin}{T_{\textnormal{min}}}
    \newcommand{\Etot}{E_{\textnormal{tot}}}

\title{Asymptotic Relaxation of Moment Equations for a Multi-species, Homogeneous BGK Model
    \thanks{
    This manuscript has been authored, in part, by UT-Battelle, LLC, under contract DE-AC05-00OR22725 with the US Department of Energy (DOE).
    The U.S. government retains the publisher, by accepting the article for publication, acknowledges that the U.S. government retains a nonexclusive, paid-up, irrevocable, worldwide license to publish or
    reproduce the published form of this manuscript, or allow others to do so, for U.S. government purposes.
    DOE will provide public access to these results of federally sponsored research in accordance with the DOE Public Access Plan
    (httpsȷ//energy.gov/downloads/doe-public-access-plan).
    \\
    This work was supported, in part, by the U.S. Department of Energy through the Los Alamos National Laboratory. Los Alamos National Laboratory is operated by Triad National Security, LLC, for the National Nuclear Security Administration of U.S. Department of Energy (Contract No. 89233218CNA000001).
    }
}

\author{
Evan Habbershaw\thanks{
    Department of Mathematics, The University of Tennessee,  Knoxville, TN, 37996
  },
Ryan S. Glasby\thanks{
    Computational Sciences and Engineering Division, Oak Ridge National Laboratory, Oak Ridge, TN, 37831
  },
Jeffrey R. Haack\thanks{
    Computational Physics and Methods Group, Los Alamos National Laboratory, Los Alamos, NM, 87545
  },
Cory D. Hauck\thanks{
    Department of Mathematics, The University of Tennessee,  Knoxville, TN, 37996
    and
    Computer Science and Mathematics Division,
Oak Ridge National Laboratory, Oak Ridge, TN, 37831
  }, and
Steven M. Wise\thanks{
  Department of Mathematics, The University of Tennessee,  Knoxville, TN, 37996
  }
}

    \date{13 September 2023}

\begin{document}
    
    \maketitle
    
    \begin{abstract}
Multi-species BGK models describe the dynamics of rarefied gases with constituent particles of different elements or compounds with potentially non-trivial velocity distributions.  In this paper, moment equations for the bulk velocities, energies, and temperatures of a spatially homogeneous multi-species BGK model are examined.  A key challenge in analyzing these equations is the fact that the collision frequencies are allowed to depend on the species temperatures, which allows for more realistic simulations of dilute gas flow.  Therefore, a positive lower bound is established for the species temperatures.  With this lower bound, a global existence and uniqueness of solutions to the coupled velocity-energy ODE system is established.  The lower bound also enables a proof of exponential decay to a unique steady-state solution.  Numerical results are presented to demonstrate how the bulk velocities and temperatures relax for large times.
    \end{abstract}

    \textbf{Key words:} Multispecies BGK, Moment Equations, Kinetic Theory, Rarefied Gas Dynamics 

    \textbf{MSC Codes:} 37A60, 34A12, 82C40

    \section{Introduction}

The Bhatnagar-Gross-Krook (BGK) equation \cite{bhatnagar} is a well-known and frequently employed model for simulating dilute gases.
It has been used in a variety of applications, including
hypersonics
\cite{baranger2019bgk,evans2017aerodynamic},
micro-channel flows
\cite{ganapol20191d,liu2018rarefaction},
nozzle flows
\cite{kumar2011simulation, kumar2010assessment,pfeiffer2019coupled},
and plasma transport
\cite{beckwith2018modelling, haack2017interfacial}.

The BGK equation prescribes the evolution of a position-velocity phase-space distribution that characterizes the state of a dilute gas and is typically employed as a simpler alternative to the Boltzmann equation.
The BGK equation can be derived via an ad-hoc approximation of the Boltzmann collision operator by a much simpler nonlinear relaxation operator \cite[Section 4.2]{liboff2003kinetic}, where the rate of relaxation is determined by a scalar collision frequency.
Like the Boltzmann collision operator, the BGK relaxation operator conserves mass, momentum, and energy.
It also dissipates entropy, and satisfies an H-Theorem that characterizes local thermodynamic equilibrium;  see e.g. \cite[Section 3.6.1]{struchtrup2005macroscopic}.

Recently, multi-species versions of the BGK equation have been derived that satisfy the conservation and entropy dissipation properties of the multi-species Boltzmann equation \cite{bobylev2018general,klingenberg2017consistent,Haack2017,haack2021consistent}.
In most of these cases, the collision frequencies are assumed to be independent of the phase-space velocity (\cite{haack2021consistent} being an exception) and dependent on the kinetic distributions only through their respective number densities (\cite{Haack2017} being an exception).
Conservation of the number density implies that collision frequencies that depend only on the number densities will be constant in time.
However, for general collision models, these frequencies will often depend on the species temperature or on the phase-space velocity \cite{book:Bird_DSMC}, and relaxing these commonly used assumptions can significantly alter the behavior of solutions \cite{haack2021consistent}.
In addition, the analysis and simulation of the BGK equations becomes more difficult when the collision frequencies depend on time and/or phase-space velocity.

In the current paper, we study the space-homogeneous, multi-species BGK equations from \cite{Haack2017}, which can be understood as a special, illustrative case of the equations in \cite{bobylev2018general} or \cite{klingenberg2017consistent}.  
We do so in the context of collision frequencies that are independent of the phase-space velocity, but depend on both number density and temperature.  
In particular, we analyze the ODE system for the momentum and energy moments (the number density evolution is trivial), which is important for asymptotic analysis and for the development of numerical methods. 
Indeed, many numerical simulations of the full BGK system (including phase-space advection) rely on implicit-explicit (IMEX) methods that treat the space-homogeneous component implicitly.  
In this setting, the result of solving the associated moment equations first is to effectively linearize and diagonalize (with respect to the phase-space velocity) the required implicit solve, yielding a significant reduction in computational cost.
(See \cite{coron,pieraccini2007implicit} for applications to the single species case and \cite[Section 4]{puppo2019kinetic} for an extension to the multi-species setting when the collision frequencies depend only on the number densities.)

In this paper, we first establish existence and uniqueness of solutions to the momentum-energy ODE system, which follows from standard theory once a lower positive bound can be established on the species temperatures.
We show monoticity of the minimum temperature envelope as well as  upper and lower bounds on the bulk velocity.
We then prove exponential decay of the species momenta and energies to their steady state values, thereby generalizing some of the results in \cite{crestetto2020kinetic} to a system with an arbitrary number of species.

The remainder of the paper is organized as follows.
In \Cref{section:modelEquations}, we introduce the basic equations.
In \Cref{section:existenceUniqueness}, we prove existence and uniqueness results and show that the species temperatures remain positive if initially so.  We also establish bounds on the bulk velocities.
In \Cref{section:asymptoticBehavior}, we establish the main results about the long-time behavior of the moment equations.
In \Cref{section:numericalDemonstration}, we provide some examples to demonstrate this behavior numerically.
Conclusions are provided in \Cref{sec:conclusions}.
The Appendix contains a statement of the Picard-Lindel\"of Theorem, which is used for the existence/uniqueness result, as well as several technical results that are used in the proofs in 
\Cref{section:asymptoticBehavior}.

    \section{Model Equations}
    \label{section:modelEquations}

We denote by $f_i$ the kinetic distribution of particles of species $i\in\{1,\cdots,N\}$ having mass $m_i$.
More specifically, $f_i(\bfx,\bfv,t)$ is the density of species $i$ particles at the point $\bfx\in\Omega\subset\R^d$, with microscopic velocity $\bfv\in\R^d$, at time $t\geq0,$ with respect to the measure $\diff \bfv \, \diff\bfx$. 
Associated to each $f_i$ are the species number density $n_i$, mass density $\rho_i$, bulk velocity $\bfu_i$, and temperature $T_i$, defined by
    \begin{equation}
\rho_i = m_in_i = m_i\int_{\R^d}f_i\diff\bfv
    ,\qquad
\bfu_i = \frac{1}{\rho_i}\int_{\R^d}m_i\bfv f_i\diff\bfv
    ,\qquad
T_i = \frac{m_i}{n_i d}\int_{\R^d}\left|\bfv-\bfu_i\right|^2f_i\diff\bfv. \label{EQ:speciesMoments}
    \end{equation}
    The species momentum densities $\bfk_i$ and species energy densities $E_i$ are given by
    \begin{equation}
\bfk_i = \rho_i \bfu_i
    \qquad \text{and} \qquad
E_i = \frac{1}{2}\rho_i|\bfu_i|^2 + \frac{d}{2}n_i T_i.
    \label{eqn:energy-def}
    \end{equation}
BGK models are expressed in terms of Maxwellian distributions that depend on these moments.

    \begin{defn}
Given $n>0,\bfu\in\R^d,\theta>0,$ a Maxwellian $M_{n,\bfu,\theta}$ is a function of the form
    \begin{align}
    \label{eq:Maxwellian_def}
M_{n,\bfu,\theta}(\bfv) = \frac{n}{(2\pi\theta)^\frac{d}{2}}\exp\left(-\frac{\left|\bfv-\bfu\right|^2}{2\theta}\right).
    \end{align}
    \end{defn}

Straightforward computations show that the moments of a Maxwellian satisfy
    \begin{align}
\int_{\R^d}M_{n,\bfu,\theta}\,\diff\bfv = n,
    \qquad 
\int_{\R^d}\bfv M_{n,\bfu,\theta} \, \diff\bfv = n\bfu,
    \qquad   
\int_{\R^d}|\bfv|^2M_{n,\bfu,\theta} \, \diff\bfv = n|\bfu|^2+dn\theta.
    \end{align}

We consider in this paper the multi-species BGK equation from \cite{Haack2017}, given by
    \begin{align}
    \label{EQ:BGKequation}
\frac{\partial f_i}{\partial t} + \bfv\cdot\nabla_{\bfx}f_i= \frac{1}{\varepsilon}\sum_j\lambda_{i,j}(M_{i,j}-f_i),
\quad \forall i\in\{1,\cdots,N\},
    \end{align}
where $\varepsilon>0$ is the Knudsen number, $\lambda_{i,j}>0$ is the (nondimensional) frequency of collisions between species $i$ and $j$ (independent of $\bfv$), and
$ M_{i,j}(\bfv) =M_{n_i,\bfu_{i,j},T_{i,j}/m_i}(\bfv) $
is a Maxwellian defined by \eqref{eq:Maxwellian_def}, using the mixture velocities and temperatures
    \begin{subequations}\label{EQ:MomentDefs}
    \begin{align}
\bfu_{i,j} &= \frac{\rho_i\lambda_{i,j}\bfu_i + \rho_j\lambda_{j,i}\bfu_j}{\rho_i\lambda_{i,j} + \rho_j\lambda_{j,i}},
    \\
T_{i,j} &= \frac{n_i\lambda_{i,j}T_i + n_j\lambda_{j,i}T_j}{n_i\lambda_{i,j} + n_j\lambda_{j,i}} + \frac{1}{d}\frac{\rho_i\rho_j\lambda_{i,j}\lambda_{j,i}}{\rho_i\lambda_{i,j}+\rho_j\lambda_{j,i}}\frac{|\bfu_i-\bfu_j|^2}{n_i\lambda_{i,j}+n_j\lambda_{j,i}}.
    \end{align}
    \end{subequations}
The existence and uniqueness of nonnegative mild solutions to (a more general) multi-species BGK equation was proved in \cite{klingenberg_2018} for periodic spatial domains and collision frequencies that depend on the number densities.

\subsection{Model Properties}

The BGK collision operators in \eqref{EQ:BGKequation} satisfy the following invariance properties:
For any $i,j,$
    \begin{align}
\int_{\R^d}\lambda_{i,j}(M_{i,j}-f_i)\diff\bfv
    &= 0,
    \\
\int_{\R^d}\lambda_{i,j}(M_{i,j}-f_i)m_i\bfv\diff\bfv + \int_{\R^d}\lambda_{j,i}(M_{j,i}-f_j)m_j\bfv\diff\bfv
    &= \mathbf{0},
    \\ 
\int_{\R^d}\lambda_{i,j}(M_{i,j}-f_i)m_i|\bfv|^2\diff\bfv + \int_{\R^d}\lambda_{j,i}(M_{j,i}-f_j)m_j|\bfv|^2\diff\bfv
    &= 0, 
    \end{align}
corresponding to conservation of species mass, total momentum, and total energy, respectively.
Furthermore, an entropy dissipation condition is satisfied:
For a spatially homogeneous mixture, the entropy function $\mathcal{H}= \sum_i\int_{\R^d}f_i\log f_i\diff\bfv$ satisfies $\frac{\diff}{\diff t}\mathcal{H} \leq 0$.
Moreover, $\frac{\diff}{\diff t}\mathcal{H} = 0$ if and only if
  $ f_i= M_{n_i,\bfu_{\textnormal{eq}},T_{\textnormal{eq}}/m_i}(\bfv) $
for all
  $i\in\{1,\cdots,N\}, $
where $\bfu_{\text{eq}}$ and $T_{\text{eq}}$ are the equilibrium bulk velocity and temperature, respectively, which are common to all species \cite{Haack2017}.

The collision frequencies $\lambda_{i,j}$ are typically expressed as functions of the species moments given in \eqref{EQ:speciesMoments}. 
To define these collision frequencies, \cite{Haack2017} presents a recipe based on matching either momentum or energy relaxation rates of the Boltzmann equation, given either a differential cross section or a momentum transfer cross section. 
In this paper, we consider the case of hard sphere collisions, for which the momentum transfer cross section is independent of microscopic velocity and given by~\cite{book:Bird_DSMC}:
    \begin{equation}
\sigma_{\textnormal{MT}} = \pi \left(\frac{d_i + d_j}{2}\right)^2 = \pi d_{i,j}^2,
    \end{equation}
where $d_i$ and $d_j$ are reference diameters for the particles of species $i$ and $j$; reference diameters for several species are given in \cite{book:Bird_DSMC} and \cite{book:MathTheory_NonUniformGases_Chapman_Cowling}. 
Following the recipe for matching energy relaxation rates from Section 4.3 of \cite{Haack2017}, we obtain the following collision frequencies for hard sphere interactions:
    \begin{equation}
\lambda_{i,j}^{\textnormal{HS}} = \frac{32 \pi^2}{3(2\pi)^{3/2}} \frac{m_i m_j}{(m_i + m_j)^2}(d_i + d_j)^2n_j\sqrt{\frac{T_i}{m_i} + \frac{T_j}{m_j}}.
    \label{eqn:collision-freqs-hhm}
    \end{equation}
While the formula in \eqref{eqn:collision-freqs-hhm} is specific to the physically relevant case $d=3$, collision frequencies can be derived in arbitrary dimensions in a similar fashion.

The dependence of the collision frequencies on the temperature makes the BGK model considered here more complicated than many models considered in the literature, which assume that $\lambda_{i,j}$ is a constant or depends only on $n_j$.
This is largely due to the non-Lipschitz nature of the square root function near zero and the fact that the collision frequencies are time-dependent functions.
While we focus on the hard spheres model in this paper for simplicity, the analysis provided here is readily applied to other models in the literature \cite{book:Bird_DSMC}.  

    \subsection{Moment Equations}
    \label{section:momentEquations}
The focus of the present work is the space homogeneous multi-species BGK equation, obtained by removing the advection terms from \eqref{EQ:BGKequation}.
In this case, $f_i = f_i(\bfv,t)$ satisfies
    \begin{align}\label{EQ:Hom_MS_BGK_Eq}
\frac{\partial f_i}{\partial t} = \frac{1}{\varepsilon}\sum_j\lambda_{i,j}(M_{i,j}-f_i), 
    \quad
\forall i\in\{1,\cdots,N\}.
    \end{align}
Multiplication of \eqref{EQ:Hom_MS_BGK_Eq} by $m_i$, $m_i\bfv$, and $m_i|\bfv|^2$, followed by integration with respect to $\bfv$, gives equations for the mass density, momentum density, and energy density of each species
    \begin{subequations}
    \label{EQ:velocity-energyODESystem1} 
    \begin{align}
\frac{\diff \rho_i}{\diff t} &= 0,
\label{eq:ConservationSpeciesMass}
    \\
\frac{\diff(\rho_i\bfu_i)}{\diff t} &= \frac{1}{\varepsilon}\sum_jA_{i,j}(\bfu_j-\bfu_i),
    \label{EQ:velocityODESystem1}
    \\
\frac{\diff E_i}{\diff t} &= \frac{1}{\varepsilon}\sum_jB_{i,j}\left(\frac{E_j}{n_j}-\frac{E_i}{n_i}\right) + \frac{1}{2\varepsilon}\sum_jB_{i,j}S_{i,j}(m_i-m_j),
    \label{EQ:energyODESystem_1}
    \end{align}
    \end{subequations}
where the equations for $\rho_i \bfu_i$ and $E_i$ are derived using the definitions for  $\bfu_{i,j}$ and $T_{i,j}$ from \eqref{EQ:MomentDefs}, and the ${N\times N}$ matrices $A$, $B$, $C$, $D$, $F$, $G$, and $S$ are given by
    \begin{subequations}
    \label{EQ:matrixDefinitions}
    \begin{gather}
\left[A\right]_{i,j} 
 = \frac{\rho_i\rho_j\lambda_{i,j}\lambda_{j,i}}{\rho_i\lambda_{i,j}+\rho_j\lambda_{j,i}}, 
    \;\;
\left[B\right]_{i,j} 
 = \frac{n_in_j\lambda_{i,j}\lambda_{j,i}}{n_i\lambda_{i,j}+n_j\lambda_{j,i}} , 
    \;\;
\left[S\right]_{i,j} 
 = \left| \bfu_{i,j}\right|^2,
    \;\;
\left[C\right]_{i,j} 
 = B_{i,j}S_{i,j}  
 ,
    \label{eq:ABSC_defs}
    \\
\left[D\right]_{i,j}
 =\Big(\sum_kA_{i,k} \Big)\delta_{i,j},
    \quad
\left[F\right]_{i,j} 
 = \Big(\sum_kB_{i,k}\Big)\delta_{i,j},
    \quad 
\left[G\right]_{i,j}
 = \Big(\sum_kC_{i,k}\Big)\delta_{i,j}.
    \label{eq:DFG_defs}
    \end{gather}
    \end{subequations}
According to \eqref{eq:ConservationSpeciesMass}, each individual species mass density $\rho_i$ (or, equivalently, number density $n_i$) is conserved (i.e., constant in time).
However, only the \emph{total} momentum and energy are conserved.
Indeed, summing \eqref{EQ:velocityODESystem1} and \eqref{EQ:energyODESystem_1} over all species yields:
    \begin{align}
\sum_{i=1}^N\frac{\diff(\rho_i\bfu_i)}{\diff t} &= \mathbf{0}
    \qquad \text{and} \qquad
\sum_{i=1}^N\frac{\diff E_i}{\diff t} = 0, 
    \label{EQ:PDEConservationLaws}
    \end{align}
respectively. 

In terms of the vectorized quantities 
    \begin{equation}
\bfU = \left[\bfu_1,\dots,\bfu_N\right]^\top \in\R^{N\times d}
    \qquad \text{and} \qquad
\bfE = [E_1,\dots E_N]^\top \in\R^N,
    \end{equation}
the equations in 
\eqref{EQ:velocity-energyODESystem1} 
take the form
    \begin{subequations}
    \label{EQ:Full_ODE_System}
    \begin{align}
P\frac{\diff\bfU}{\diff t} &= -\frac{1}{\varepsilon}(D-A)\bfU,
    \label{EQ:velocityODESystem2} 
    \\
\frac{\diff\bfE}{\diff t} &= -\frac{1}{\varepsilon}(F-B)Q^{-1}\bfE + \frac{1}{2\varepsilon}(G-C)M\mathbf{1},
    \label{EQ:energyODESystem_2}
    \end{align}
    \end{subequations}
where 
$P=\text{diag}\{\rho_k\} \in\R^{N\times N}$, 
$Q=\text{diag}\{n_k\}\in\R^{N\times N}$,
and $M=\text{diag}\{m_k\} \in\R^{N\times N}$.
To simplify the analysis later, it will be convenient to introduce variables
    \begin{equation}
\bfW=P^\frac{1}{2}\bfU\in\R^{N\times d} 
    \quad \text{ and } \quad
\bfxi=Q^{-\frac{1}{2}}\bfE\in\R^N. 
    \end{equation}
In terms of $\bfW$ and $\bfxi$, \eqref{EQ:Full_ODE_System} takes the form
    \begin{subequations}
    \label{EQ:Scaled_ODE_System}
    \begin{align}
\frac{\diff\bfW}{\diff t} &= - \frac{1}{\varepsilon}Z\bfW ,
    \label{EQ:WvelocityODESystem1} 
    \\
\frac{\diff\bfxi}{\diff t} &= -\frac{1}{\varepsilon}\widehat{Z}\bfxi + \frac{1}{2\varepsilon}Q^{-\frac{1}{2}}(G-C)M\mathbf{1},
    \end{align}
    \end{subequations}
where 
    \begin{align}
    \label{eq:Z_Zhat}
Z = P^{-\frac{1}{2}}(D-A)P^{-\frac{1}{2}} \in \mathbb{R}^{N\times N}
    \quad \text{ and } \quad
\widehat{Z} = Q^{-\frac{1}{2}}(F-B)Q^{-\frac{1}{2}}\in \mathbb{R}^{N\times N}.
    \end{align}

    \section{Properties of the ODE system}
    \label{section:existenceUniqueness}

The main result of this section is to show that there is a unique solution to the ODE system \eqref{EQ:Full_ODE_System} for all physically meaningful initial conditions.
When the collision frequencies are constant in time, such a result follows directly from the standard ODE theory for systems with Lipschtiz dynamics.
However, because the collision frequencies in \eqref{eqn:collision-freqs-hhm} depend on $\sqrt{T_i}$, only a local Lipschitz condition can be obtained for the right-hand side of \eqref{EQ:Full_ODE_System}.
Thus the key step is to show that the species temperatures are bounded below away from zero.
This bound will also be important for establishing exponential decay to steady-state in \Cref{section:asymptoticBehavior}.
We also establish upper and lower bounds on the bulk velocity components that will be used in the decay estimates.

    \subsection{Positivity of the Temperature}

To make the following proofs easier to manage, we introduce the parameters
    \begin{align}
\alpha_{i,j} &= \frac{\rho_i\lambda_{i,j}}{\rho_i\lambda_{i,j}+\rho_j\lambda_{j,i}}
    \qquad \text{and} \qquad
\beta_{i,j} = \frac{n_i\lambda_{i,j}}{n_i\lambda_{i,j}+n_j\lambda_{j,i}},
    \end{align}
which satisfy the conditions $\alpha_{i,j}+\alpha_{j,i}=1=\beta_{i,j}+\beta_{j,i}$.
With these parameters, the mixture values in \eqref{EQ:MomentDefs} can be expressed as
    \begin{align}
    \label{EQ:mixtureAlphaBeta}
\bfu_{i,j} = \alpha_{i,j}\bfu_i+\alpha_{j,i}\bfu_j
    \quad \text{and} \quad
T_{i,j} = \beta_{i,j}T_i+\beta_{j,i}T_j +  \frac{1}{d} m_i \alpha_{j,i} \beta_{i,j} \left|\bfu_i-\bfu_j\right|^2.
    \end{align}
    
    \begin{lemma}
    \label{lemma:simpleTemp}
The temperatures $T_i$, $i \in \{1,\dots,N\}$, satisfy the ODE
    \begin{align}
    \label{EQ:temperature_simple}
\frac{\diff T_i}{\diff t}
  &= \frac{1}{\varepsilon}\sum_j\lambda_{i,j}\beta_{j,i}(T_j-T_i) + \frac{1}{\varepsilon d}\sum_j\lambda_{i,j}m_i\alpha_{j,i}(\alpha_{j,i}+\beta_{i,j})|\bfu_i-\bfu_j|^2.
    \end{align}
    \end{lemma}
    
    \begin{proof}
The proof is a direct calculation.
From  \eqref{EQ:velocityODESystem1} and the fact that $A_{i,j} = \alpha_{j,i}\rho_i\lambda_{i,j}$, it follows that
    \begin{equation}
\rho_i\frac{\diff}{\diff t}|\bfu_i|^2
  = \frac{1}{\varepsilon}\sum_j\rho_i\lambda_{i,j}(2\alpha_{j,i}\langle\bfu_i,\bfu_j\rangle - 2\alpha_{j,i}|\bfu_i|^2),
    \label{EQ:s_iDeriv}
    \end{equation}
where $\langle\cdot,\cdot\rangle$ denotes the standard Euclidean inner product on $\R^d$.
The equation in \eqref{EQ:temperature_simple} can be obtained by
(i) using \eqref{eqn:energy-def} to express \eqref{EQ:energyODESystem_1} in terms of the temperatures and velocities;
(ii) invoking \eqref{EQ:s_iDeriv} to eliminate the time derivative of $|\bfu_i|^2$;
(iii) using the formula in \eqref{EQ:mixtureAlphaBeta} to replace $\bfu_{i,j}$ wherever it appears;
and (iv) applying the elementary relations
$\alpha_{j,i}^2 = \alpha_{j,i}- \alpha_{i,j}\alpha_{j,i} = 2\alpha_{j,i}+\alpha_{i,j}^2-1$.
    \end{proof}

    \begin{theorem}
    \label{thm:temperatureBoundedBelow}
Suppose that for some $t_{\textnormal{f}} >0$, there exists a local solution $\bfU\in C^1\left([0,t_{\textnormal{f}}]; \mathbb{R}^{N\times d}\right)$
and $\bfE\in C^1\left([0,t_{\textnormal{f}}];\mathbb{R}^N\right)$ to the system given by \eqref{EQ:Full_ODE_System}, with initial conditions $\bfU(0)$ and $\bfE(0)$ such that $T_i(0) > 0$ for all $i \in \{1,\cdots, N\}$. 
Then $T_i(t)\geq T_{\min}\coloneqq 
\min_k\{T_k(0)\}$ for all $t\in[0,t_{\textnormal{f}}]$ and all $i\in\{1,\cdots,N\}$.
    \end{theorem}
    
    \begin{proof}
The proof relies on an integrating factor technique that yields a Gr\"onwall-type estimate.
Define a lower temperature envelope
$T_\star(t) = \min_{k}\{T_k(t)\}$.
From \eqref{EQ:temperature_simple} it follows that, for $i \in \{1,\cdots,N\}$,
    \begin{align}
\frac{\diff T_i}{\diff t} + c_i(t)T_i(t) &\geq c_i(t)T_\star(t),
    \quad \text{ where }\quad
c_i(t) = \frac{1}{\varepsilon}\sum_j\lambda_{ij}(t)\beta_{ji}(t).
    \label{EQ:diff_ineq_1}
    \end{align}
Let $a_i(t) = \int_0^tc_i(s)\diff s$.
Then $a'_i(t)=c_i(t)$ and 
    \begin{align}
\frac{\diff}{\diff t}\left(e^{a_i(t)}T_i(t)\right)
    \overset{\eqref{EQ:diff_ineq_1}}{\geq} e^{a_i(t)}c_i(t)T_\star(t)
    = \frac{\diff}{\diff t}\left(e^{a_i(t)}\right)T_\star(t).
    \label{EQ:diff_ineq_2}
    \end{align}
For each $\gamma>0$, define the set 
    \begin{align}
S_\gamma = \left\{t \in [0,t_{\textnormal{f}}]\;\middle|\;T_\star(t)<\frac{1}{1+\gamma}T_\star(0)\right\}.
    \end{align}
We now argue by contradiction.
Assume that $S_\gamma$ is nonempty and let $t_\star = \inf S_\gamma$.
Then by definition,
    \begin{align}
T_\star(t_\star) \leq \frac{1}{1+\gamma}T_\star(0) < T_\star(0). \label{EQ:ineq_2}
    \end{align}
However, because $T_i$ is continuous on $[0,t_{\textnormal{f}}]$, it follows that
$T_\star(t)\geq T_\star(t_\star)$ for every $t\in[0,t_\star]$, which, along with \eqref{EQ:diff_ineq_2}, implies that
    \begin{align}
    \label{EQ:diff_ineq_3}
\frac{\diff}{\diff t}\left(e^{a_i(t)}T_i(t)\right) \geq \frac{\diff}{\diff t}\left(e^{a_i(t)}\right)T_\star(t_\star),
    \qquad \forall t\in[0,t_\star].
    \end{align}
Integrate both sides of \eqref{EQ:diff_ineq_3} on $[0,t_\star]$, multiply the result by $e^{-a_i(t_\star)} \geq 0$, and use $T_i(0)\geq T_{\star}(0)$, to get
    \begin{equation}
    \begin{split}
T_i(t_\star)
  &\geq T_i(0)e^{-a_i(t_\star)} + T_\star(t_\star)\left(1-e^{-a_i(t_\star)}\right)
    \\
  &\geq T_\star(0) e^{-a_i(t_\star)} + T_\star(t_\star)\left(1-e^{-a_i(t_\star)}\right).
    \label{EQ:int_ineq_1}
    \end{split}
    \end{equation}
Define $i_\star$ such that $T_{i_\star}(t_\star) = T_\star(t_\star)$.
If $i = i_{\star}$, then  \eqref{EQ:int_ineq_1} becomes
    \begin{align}
T_\star(t_\star) \equiv T_{i_\star}(t_\star) \geq T_\star(0)e^{-a_{i_\star}(t_\star)} + T_\star(t_\star)\left(1-e^{-a_{i_\star}(t_\star)}\right),
    \label{EQ:ineq_1}
    \end{align}
which, after some simple algebra, implies that $T_\star(t_\star) \geq T_\star(0)$.
However this results contradicts 
\eqref{EQ:ineq_2}
which means that no such $t_\star$ exists, and the set $S_\gamma$ must be empty.
As a consequence, $T_\star(t)\geq\frac{1}{1+\gamma}T_\star(0)$ for all $t \in [0, t_{\textnormal{f}}]$.
Since $\gamma > 0$ is arbitrary, the limit $\gamma\to0$ yields $T_\star(t) \geq T_\star(0)=T_{\textnormal{min}}$ for all $t \in [0, t_{\textnormal{f}}]$.
The proof is complete.
    \end{proof}

    \subsection{Velocity Bounds}

For each component $k\in\{1,\cdots,d\}$, define the lower and upper velocity envelopes
    \begin{align}
\underline{u}_k(t) = \min_j\{U_{j,k}(t)\},
    \qquad
\overline{u}_k(t) = \max_j\{U_{j,k}(t)\},
    \end{align}
    and the components of the vector $\bfu_{\max}\in\R^d$ 
    \begin{align}
\left[\bfu_{\max}(t)\right]_{k} &= \max\left\{\left|\underline{u}_k(t)\right|,\left|\overline{u}_k(t)\right|\right\}.
    \end{align}
    
    \begin{theorem}
    \label{thm:velocityBounds}
Suppose that the assumptions of \Cref{thm:temperatureBoundedBelow} hold.
Then for each $k\in\{1,\cdots,d\}$, the velocity components follow the inequalities
    \begin{align}\label{EQ:velocity_inequalities}
\underline{u}_k(0)
    \leq
\underline{u}_k(t)
    \leq
U_{i,k}(t)
    \leq
\overline{u}_k(t)
    \leq
\overline{u}_k(0)
    \end{align}
for all $t\in[0,t_{\textnormal{f}}]$ and all $i\in\{1,\cdots,N\}$.
Further, for any $t\in[0,t_{\textnormal{f}}]$, and all $i\in\{1,\cdots,N\}$,
    \begin{align}
    \label{eq:u_max}
\left\|\bfu_i(t)\right\|_2 &\leq \left\|\bfu_{\max}(t)\right\|_2 \leq \left\|\bfu_{\max}(0)\right\|_2 \eqqcolon u_{\max}.
    \end{align}
    \end{theorem}
    
    \begin{proof}
The proofs of the inequalities in \eqref{EQ:velocity_inequalities} are similar to the proof of the temperature lower bound given in \Cref{thm:temperatureBoundedBelow}.
From the velocity equation, \eqref{EQ:velocityODESystem1},
    \begin{align}
\frac{\diff}{\diff t}U_{i,k} + \left(\frac{1}{\varepsilon}\sum_j\lambda_{i,j}\alpha_{j,i}\right)U_{i,k} &= \frac{1}{\varepsilon}\sum_j\lambda_{i,j}\alpha_{j,i}U_{j,k},
    \end{align}
we obtain, using the definitions
$U_{j,k}(t)\leq \overline{u}_k(t)$, and $c_i = \frac{1}{\varepsilon}\sum_j\lambda_{i,j}\alpha_{j,i}$,
    \begin{align}
\frac{\diff}{\diff t}U_{i,k} + c_i(t)U_{i,k}
  \leq c_i(t)\overline{u}_k(t).
    \label{EQ:velProof_ineq1}
    \end{align}
From this point the proof of the upper bound in \eqref{EQ:velocity_inequalities} follows that of \Cref{thm:temperatureBoundedBelow}, with \eqref{EQ:velProof_ineq1} being the  analog of \eqref{EQ:diff_ineq_1}.
The proof of the lower bound is even more similar.
    
Finally, the bound in \eqref{eq:u_max} follows from \eqref{EQ:velocity_inequalities}.
Specifically, if $a$, $b$, and $c$ are real-valued scalars such that $a\leq b\leq c$, then $|b|\leq\max\{|a|,|c|\}$.
Hence
    \begin{align}
\left\|\bfu_i\right\|_2^2 
  =\sum_{k=1}^dU_{i,k}^2
  \leq \sum_{k=1}^d\max\left\{\left|\underline{u}_k(0)\right|,\left|\overline{u}_k(0)\right|\right\}^2
  = \left\|\bfu_{\max}(0)\right\|_2^2 = u_{\max}^2.
    \end{align}
    \end{proof}

    \subsection{Existence and Uniqueness}

According to \eqref{eq:ConservationSpeciesMass}, the number densities $n_i$ and mass densities $\rho_i$ are non-negative constants in time:
    \begin{equation}
n_i(0) = n_i(t), \quad \rho_i(0) = \rho_i(t) = m_i n_i(t), \quad \forall \, i\in\{1,\cdots , N\},
    \quad t\in [0,t_{\textnormal{f}}].
    \end{equation}
To avoid any degenerate cases, we assume that $n_{\textnormal{min}} \coloneqq \min_j n_j(0) > 0$ and $\rho_{\textnormal{min}} \coloneqq \min_j \rho_j(0) > 0$.
Then, given  any velocity-energy pair $(\bfu_i,E_i)$, the associated temperature $T_i$ is given by (cf. \eqref{eqn:energy-def})
    \begin{equation}
    \label{eq:vartheta_defn}
T_i = \vartheta_i(\bfu_i,E_i) \coloneqq \frac{2}{dn_i}E_i-\frac{m_i}{d}|\bfu_i|^2,
    \end{equation}
and the set of all realizable velocity and energy states is given by
    \begin{equation}
    \label{eq:R_defn}
\mathcal{R} 
= \left\{ ({\bfU},\bfE)\in\R^{N\times d}\times \R^N
\;  \Big | \;  \vartheta_i(\bfu_i,E_i) \geq 0, \;  \forall \, i\in\{1,\cdots , N\} 
\right \},
    \end{equation}
where $\vartheta_i$ is defined in \eqref{eq:vartheta_defn}.
Given positive scalars $\Tmin \in (0,\infty)$ and $\Etot \in (0,\infty)$, let
    \begin{equation}
\mathcal{D}(\Tmin, \Etot) 
= \left\{ ({\bfU},\bfE)\in
\mathcal{R}
\; \Bigg | \; 
\sum_{i=1}^N E_i 
= \Etot 
,\;
\vartheta_i(\bfu_i,E_i) \geq \Tmin, \; \forall \, i\in\{1,\cdots , N\} 
\right\}
    \end{equation}
and recall that by \Cref{thm:velocityBounds},
$|\bfu_i|\leq u_{\max}\eqqcolon\left\|\bfu_{\max}(0)\right\|_2\leq \sqrt{\frac{2\Etot}{ \rho_{\textnormal{min}}}}$
for all $(\bfU,\bfE) \in \mathcal{D}(\Tmin, \Etot)$.
Thus $\mathcal{D}(\Tmin, \Etot)$ is a closed and bounded subset of $\R^{N\times d}\times\R^N$.

    \begin{lemma}
    \label{lemma:setBoundInvarianceLemma}
Let $\Tmin \in (0,\infty)$ and $\Etot \in (0,\infty)$ be given.
Suppose that for some $t_{\textnormal{f}} >0$, there exists a local solution $\bfU\in C^1\left([0,t_{\textnormal{f}}];\mathbb{R}^{N\times d}\right)$ and $\bfE\in C^1\left([0,t_{\textnormal{f}}];\mathbb{R}^N\right)$ of the system in \eqref{EQ:Full_ODE_System} with initial condition $(\bfU^0, \bfE^0) \in \mathcal{D}(\Tmin, \Etot)$.
Then for all $t\in[0,t_{\textnormal{f}}],$ $(\bfU(t),\bfE(t)) \in \mathcal{D}(\Tmin, \Etot)$.
In particular, $\Tmin \leq T_i(t) \leq \Tmax\coloneqq\frac{2}{d n_{\textnormal{min}}} \Etot$ for all $t\in [0,t_{\textnormal{f}}]$.
    \end{lemma}
    
    \begin{proof}
To conclude that $(\bfU(t),\bfE(t)) \in \mathcal{D}$ for all $t\in [0,t_{\textnormal{f}}]$, two conditions must be satisfied.
The first condition:
    \begin{equation}
\sum_{i=1}^NE_i(t) = \sum_{i=1}^NE_i(0) = \Etot, \quad \forall \, t\in [0,t_{\textnormal{f}}],
    \end{equation} 
follows immediately from \eqref{EQ:PDEConservationLaws}, and the upper bound $T_i(t) \leq \Tmax$ follows:
    \begin{equation}
T_i(t)
    = \frac{2}{d n_i}E_i-\frac{m_i}{d}|\bfu_i|^2
    \leq \frac{2}{d n_{\textnormal{min}}} \Etot
    = \Tmax.
    \end{equation}
The second condition:
    \begin{equation}
\Tmin
    \leq T_i(t)
    =\vartheta(\bfu_i(t),E_i(t)),
    \quad \forall \, t\in [0,t_{\textnormal{f}}],
    \end{equation}
is a direct consequence of \Cref{thm:temperatureBoundedBelow}.
    \end{proof}
    
    \begin{theorem}
    \label{theorem:existenceAndUniqueness}
Suppose that $m_i >0 $ and $n_i>0$ for all $i \in \{1, \cdots, N\}$.
Then for any $(\bfU^0,\bfE^0) \in \operatorname{int} \mathcal{R}$, there exists a global, unique solution $(\bfU,\bfE) \in C^1([0,\infty);\mathcal{R})$ of the system \eqref{EQ:Full_ODE_System} with the initial conditions $(\bfU(0),\bfE(0)) = (\bfU^0,\bfE^0)$.
Moreover, the associated temperatures $T_i(t) = \vartheta_i(\bfu_i(t),E_i(t))$ are bounded below by their initial values; that is 
    \begin{equation}
\min_i T_i(t) \geq \min_i T_i(0),
    \qquad 
\forall t > 0.
    \end{equation}
    \end{theorem}

    \begin{proof}
The system  \eqref{EQ:Full_ODE_System} can be written as
    \begin{equation}
\frac{\diff}{\diff t}
  \begin{pmatrix}
    \bfU \\ \bfE
  \end{pmatrix} 
 = \bff(\bfU,\bfE)
    \quad \text{where} \quad 
\bff(\bfU,\bfE) \coloneqq
  \begin{pmatrix}
    -\frac{1}{\varepsilon}P^{-1}(D-A)\bfU 
  \\
    -\frac{1}{\varepsilon}(F-B)Q^{-1}\bfE + \frac{1}{2\varepsilon}(G-C)M\mathbf{1} 
  \end{pmatrix}
.
    \end{equation}
Let $\Tmin = \min_i T_i(0)$ and $\Etot = \sum_i E_i(0)$.
Since $(\bfU,\bfE) \in \operatorname{int} \mathcal{R}$, it follows that $\Tmin > 0$ and $\Etot > 0$.
Moreover, there exists an $\epsilon>0$ such that the closed $\epsilon$-neighborhood of $\mathcal{D} (\Tmin, \Etot)$, denoted by $\mathcal{D}_\epsilon (\Tmin, \Etot)$,  is contained in $\operatorname{int}\mathcal{R}$, i.e.,  $\mathcal{D}_\epsilon (\Tmin, \Etot) \subset \operatorname{int}\mathcal{R}$.
In particular,
$  \min_i T_i = \min_i \vartheta(\bfu_i,E_i). $
is bounded below on $\mathcal{D}_\epsilon (\Tmin, \Etot)$.
Hence, $f$ is Lipschitz on $\mathcal{D}_\epsilon(\Tmin, \Etot)$ with Lipschitz constant $L>0$ and bound $\|\bff(\bfU,\bfE)\|_2 \le \mathcal{M} < \infty$.
Since $(\bfU^0,\bfE^0)\in\mathcal{D}$, there is a number $\beta = \beta(\epsilon)>0$, independent of the point $(\bfU^0,\bfE^0)$, such that the closed ball $\overline{B}\left((\bfU,\bfE),\beta\right)\subset \mathcal{D}_\epsilon$.
Appealing to the Picard-Lindelöf Theorem~\ref{thm:picard-lindelof}, there exists a unique solution on the interval $[0,\delta]$, where $\delta = \min\left(\frac{1}{2L},\frac{\beta}{\mathcal{M}}\right)$.
By \Cref{lemma:setBoundInvarianceLemma}, $(\bfU(t),\bfE(t))\in \mathcal{D}$, for all $t\in [0,\delta]$.
Since $(\bfU(\delta),\bfE(\delta))\in \mathcal{D}$, we can apply the Picard-Lindel\"{o}f Theorem to \eqref{EQ:Full_ODE_System} again, but with the initial condition $ = (\bfU(\delta),\bfE(\delta))$, and thereby extend the unique solution to the interval $[0,2\delta]$.
Again, $(\bfU(2\delta),\bfE(2\delta)) \in\mathcal{D}$, and the process can be continued indefinitely, with, crucially, no degradation of $\delta>0$, since $\mathcal{M}$ and $L$ cannot grow larger and $\beta$ does not shrink.
The result follows.
    \end{proof}

    \section{Long-time Behavior}
    \label{section:asymptoticBehavior}
    
Having proven the existence of unique, global, physically realizable solutions for
\eqref{EQ:Full_ODE_System}, we now seek to characterize the asymptotic behavior of those solutions.
In particular, we want to show that solutions converge to
    \begin{align}
    \label{eq:U_E_infinity_def}
\bfU^\infty(t) = \mathbf{1}\left(\bfu^{\infty}\right)^\top 
    \qquad \text{and} \qquad
\bfE^\infty(t) &= \frac{|\bfu^\infty(t)|^2}{2}P\mathbf{1} + \frac{dT^{\infty}(t)}{2}Q\mathbf{1},
    \end{align}
where
    \begin{equation}
    \label{EQ:T_infinity_Definition}
\bfu^{\infty}(t) = \frac{\sum_i\rho_i\bfu_i(t)}{\sum_i\rho_i} 
    \quad \text{and} \quad
T^{\infty}(t) = \frac{\sum_in_iT_i(t)}{\sum_in_i} + \frac{\sum_i\rho_i(|\bfu_i(t)|^2-|\bfu^{\infty}(t)|^2)}{d\sum_in_i}
    \end{equation}
are constant functions of time.

    \begin{prop}
    \label{lemma:mixtureValuesAreConstant}
Suppose that $(\bfU, \bfE) \in C^1\left([0,\infty);\mathcal{R}\right)$ is a unique global-in-time solution to the system \eqref{EQ:Full_ODE_System} with assumptions as given in \Cref{theorem:existenceAndUniqueness}. 
Further let $\bfT \in C^1\left([0,\infty);\R^N\right)$ be the associated temperature vector, whose components satisfy \eqref{EQ:temperature_simple}.
Then $\bfu^\infty$, $T^\infty$, and $\bfE^\infty$ are time invariant quantities, and $T^\infty >0$.
    \end{prop}
    
    \begin{proof}
By conservation of mass, \eqref{eq:ConservationSpeciesMass}, and total momentum, \eqref{EQ:PDEConservationLaws}, $\sum_i\rho_i$ and $\sum_i\rho_i\bfu_i$ are constant.
Thus $\bfu^{\infty}$ is time invariant.
To show that $T^{\infty}$ is constant in time, multiply \eqref{EQ:T_infinity_Definition} by the expression $\frac{d}{2}\sum_in_i$, which is constant by \eqref{eq:ConservationSpeciesMass}, to obtain
    \begin{equation}
    \begin{split}
\frac{d}{2}\left(\sum_in_i\right)T^\infty 
 &= \sum_i\left[\frac{d}{2}n_iT_i+\frac{1}{2}\rho_i|\bfu_i|^2\right] - \frac{1}{2}\left(\sum_i\rho_i\right)|\bfu^\infty|^2 \\
 &= \sum_iE_i - \frac{1}{2}\left(\sum_i\rho_i\right)|\bfu^\infty|^2. 
    \label{EQ:TMixIsConstant}
    \end{split}
    \end{equation}
The first term on the right-hand side of \eqref{EQ:TMixIsConstant} is constant by conservation of total energy, \eqref{EQ:PDEConservationLaws}; the second term is constant by conservation of species mass, \eqref{eq:ConservationSpeciesMass}, and the fact that $\bfu^{\infty}$ is constant.
Thus $T^\infty$ is time invariant.
It follows immediately from \eqref{eq:U_E_infinity_def} that $\bfE^{\infty}$ is also time invariant.

Since $(\bfU(t),\bfE(t)) \in \mathcal{R}$ for all $t>0$, it follows that $T_i(t) = \vartheta(\bfu_i(t),E_i(t)) \geq 0$ for all $t>0$.
To show that $T^\infty>0,$ it is sufficient to show that $\sum_i\rho_i(|\bfu_i|^2-|\bfu^\infty|^2)\geq0$.
Let $r_i = \frac{\rho_i}{\sum_k\rho_k}>0$.
Then $\sum_ir_i=1$ and $|\bfu^{\infty}|^2=(\sum_ir_i\bfu_i)^\top (\sum_jr_j\bfu_j)=\sum_i\sum_jr_ir_j\bfu_i^\top \bfu_j$. 
Therefore
    \begin{equation}
    \begin{split}
\sum_i\rho_i(|\bfu_i|^2-|\bfu^\infty|^2)
  &= \sum_k\rho_k\left[\sum_{i,j} r_ir_j|\bfu_i|^2 
- \sum_{i,j}  r_ir_j\bfu_i^\top \bfu_j\right] \\
  &= \frac{1}{2}\sum_k\rho_k\sum_{i,j} r_ir_j\left|\bfu_i-\bfu_j\right|^2
  \geq 0.
    \end{split}
    \end{equation}
    \end{proof}

    \begin{remark}
The quantities $\bfu^\infty$ and $T^\infty$ were referred to as the mixture mean velocity and mixture temperature, respectively, in \cite{Haack_2023}.
However, in this paper, we reserve these terms for the quantities $\bfu_{i,j}$ and $T_{i,j}$, respectively.
    \end{remark}
 
The convergence of $\bfu_i$, $T_i$, and $\bfE$ as $t \to \infty$ is established by the following result, the proof of which is the focus of the rest of this section.

    \begin{theorem}
    \label{theorem:mainResult}
Under the assumptions of \Cref{lemma:mixtureValuesAreConstant}, 
for all $i \in \{1,\cdots, N\}$,
    \begin{align}
\lim_{t\to\infty}\bfu_i(t) = \bfu^\infty, \qquad \lim_{t\to\infty}T_i(t)=T^\infty,
    \qquad \mbox{and} \qquad
\lim_{t\to\infty}\bfE(t)=\bfE^\infty,
    \end{align}
and bounds for the decay rates of the bulk velocities and energies are given by
    \begin{align}
\left\|\bfu_i(t)-\bfu^\infty\right\|_2 
&\leq C_{\bfU}e^{-\frac{z_{\min}}{\varepsilon}t} , \qquad i \in  \{1, \dots, N\}
    \\
\left\|\bfE(t)-\bfE^\infty\right\|_2 
&\leq C_1 e^{-\frac{\widehat{z}_{\min}}{\varepsilon}t}+C_2\frac{e^{-\frac{z_{\min}}{\varepsilon}t}-e^{-\frac{\widehat{z}_{\min}}{\varepsilon}t}}{\widehat{z}_{\min}-z_{\min}} ,
    \end{align}
where $z_{\min}>0$ and $\widehat{z}_{\min}>0$ are lower bounds on the positive eigenvalues of $Z$ and $\widehat{Z}$, respectively, and $C_{\bfU}>0$, $C_1>0$, and $C_2>0$ are constants depending only on the initial conditions of the system.
    \end{theorem}

    \subsection{Null Spaces of \texorpdfstring{$D-A$}{D-A}, \texorpdfstring{$Z$}{Z}, \texorpdfstring{$F-B$}{F-B}, and \texorpdfstring{$\widehat{Z}$}{ZHat}}

To prove
\Cref{theorem:mainResult}, we first characterize the null spaces of the matrices $D-A$, $Z$, $F-B$, and $\widehat{Z}$, which are defined in \eqref{EQ:matrixDefinitions}, and \eqref{eq:Z_Zhat}.

    \begin{lemma}
    \label{lemma:invariantNullSpacesZ}
The matrices $D-A$ and $Z$ are symmetric positive semi-definite (SPSD), each with a one dimensional null space.
In particular,  $\mathcal{N}(D-A) = {\textnormal{span}}(\{\mathbf{1}\})$  and $\mathcal{N}(Z) = {\textnormal{span}}\left(\left\{P^\frac{1}{2}\mathbf{1}\right\}\right)$.
Moreover, the null space and range space of $Z$ are invariant with respect to time.
    \end{lemma}

    \begin{proof}
Clearly $D-A$ and $Z$ are symmetric by inspection. 
For any $\bfy\in\R^N,$ 
    \begin{equation}
\bfy^\top (D-A)\bfy = \frac{1}{2}\sum_{i,j}A_{i,j}(y_i-y_j)^2\geq0.
    \end{equation}
Moreover, since $A_{i,j}>0$, $\frac{1}{2}\sum_{i,j}A_{i,j}(y_i-y_j)^2=0$ if and only if $y_i=y_j$ for all $i$ and $j$, in which case $\bfy=c\mathbf{1}$ for some $c \in \R$. 
Thus $D-A$ is SPSD, with a one-dimensional null space spanned by the eigenvector $\mathbf{1}$.
Similarly, $\bfy^\top Z\bfy=(P^{-\frac{1}{2}}\bfy)^\top (D-A)(P^{-\frac{1}{2}}\bfy)\geq0$, and $Z\bfy=\mathbf{0} \iff P^{-\frac{1}{2}}\bfy=c\mathbf{1}\iff\bfy=cP^\frac{1}{2}\mathbf{1}$, for some $c \in \R$.
Thus $Z$ is SPSD, with a one dimensional null space spanned by the eigenvector $P^\frac{1}{2}\mathbf{1}$.    
By conservation of mass, \eqref{eq:ConservationSpeciesMass}, $P^\frac{1}{2}\mathbf{1}$ is indepedent of time. 
Thus, the null space of $Z$ is invariant with respect to time. 
The symmetry of $Z$ implies that $\mathcal{R}(Z)=\mathcal{N}(Z)^\perp$ is also invariant.
    \end{proof}

    \begin{lemma}
    \label{lemma:invariantNullSpacesZHat}
The matrices $F-B$ and $\widehat{Z}$ are symmetric positive semi-definite (SPSD), with $\mathcal{N}(F-B) = {\textnormal{span}}\left(\left\{\mathbf{1}\right\} \right)$ and $\mathcal{N}(\widehat{Z}) = {\textnormal{span}}\left( \left\{Q^\frac{1}{2}\mathbf{1}\right\}\right)$.
Moreover, the null space and range space of $\widehat{Z}$ are invariant with respect to time.
    \end{lemma}

    \begin{proof}
The proof follows closely that of \Cref{lemma:invariantNullSpacesZ}.
Details are left to the reader.
    \end{proof}
    
    \subsection{Velocity Relaxation Proof}
    \label{section:velocityRelaxationProof}

    \begin{theorem}
    \label{theorem:velocityRelaxationProof}
Under the assumptions of \Cref{lemma:mixtureValuesAreConstant}, $\lim_{t\to\infty}\bfu_i(t) =\bfu^\infty$
and
    \begin{align}
    \label{eq:u_decay_bound}
\left\|\bfu_i(t)-\bfu^\infty\right\|_2 &\leq C_{\bfU}e^{-\frac{z_{\min}}{\varepsilon}t}, \quad i = 1,\cdots , N,
    \end{align}
where $C_{\bfU}>0$ is a constant that depends on the initial conditions of the system and $z_{\min}>0$ is a lower bound on the positive eigenvalues of $Z$.
    \end{theorem}
    
    \begin{proof}
The proof is based on \eqref{EQ:WvelocityODESystem1} which is equivalent to \eqref{EQ:velocityODESystem2}. 
Let   $\bfW^\infty=P^\frac{1}{2}\bfU^\infty$, where $\bfU^\infty = \mathbf{1} \left(\bfu^\infty\right)^\top$ is given in \eqref{eq:U_E_infinity_def}.
Then \Cref{lemma:invariantNullSpacesZ} implies that $Z\bfW^\infty=\mathbf{0}\in\R^{N\times d}$.
Moreover, since $P$ and $\bfU^\infty$ are independent of time, so too is $\bfW^\infty$.
Thus in terms of $\widetilde{\bfW}\coloneqq \bfW-\bfW^\infty$, \eqref{EQ:WvelocityODESystem1} takes the form
    \begin{align}
\frac{\diff\widetilde{\bfW}}{\diff t} &= -\frac{1}{\varepsilon}Z\widetilde{\bfW}.
    \label{EQ:W_tilde_ODE}
    \end{align}
Write the element $\widetilde{\bfW}$ as a sum of the null and range space components with respect to $Z$: $\widetilde{\bfW}=\widetilde{\bfW}_\mathcal{R}+\widetilde{\bfW}_\mathcal{N}$.
As shown in \Cref{appendix:extraProofs}, $\widetilde{\bfW}_\mathcal{N}=0$.
Hence $\widetilde{\bfW} = V \bfA$, where 
the columns of $V=[\bfv_1,\cdots,\bfv_{N-1}]\in\R^{N\times(N-1)}$ are a set of orthonormal eigenvectors of $Z$ that span $\mathcal{R}(Z)$, and $\bfA = [\bfa_1,\cdots,\bfa_{d}]\in\R^{(N-1)\times d}$.
Using this formulation, the ODE system \eqref{EQ:W_tilde_ODE} becomes
    \begin{equation}
    \label{eq:A-ODE}
\frac{\diff\bfA}{\diff t} = -\frac{1}{\varepsilon}V^\top ZV\bfA.
    \end{equation}
The eigenvalues of $V^\top ZV$ are strictly positive and can be bounded below by $z_{\min} \coloneqq \frac{N\min_{i,j}A_{i,j}}{\max\{\rho_k\}}$
(see \Cref{appendix:eigenvalueBounds}).
Thus $\left(\bfA,V^\top ZV\bfA\right)_{\text{F}}\geq z_{\min}\left\|\bfA\right\|_{\text{F}}^2$, where $\| \cdot \|_{\text{F}}$ is the Frobenius norm.
Let  $s(t)=\left\|\bfA\right\|_{\text{F}}^2$.
Then, using \eqref{eq:A-ODE},
    \begin{align}
\frac{1}{2}\frac{\diff s}{\diff t}
 =\left(\bfA,\frac{\diff\bfA}{\diff t}\right)_{\text{F}}
 = \left(\bfA,-\frac{1}{\varepsilon}V^\top ZV\bfA\right)_{\text{F}}
 \leq -\frac{z_{\min}}{\varepsilon}\left\|\bfA\right\|_{\text{F}}^2
 = -\frac{z_{\min}}{\varepsilon} s, 
    \end{align}
which implies that 
    \begin{equation}
    \label{EQ:A_inequality}
\left\|\bfA(t) \right\|_{\text{F}}^2 \leq \left\|\bfA(0)\right\|_{\text{F}}^2e^{-\frac{2z_{\min}}{\varepsilon}t}. 
    \end{equation}
Since $V^\top V=I_{N-1}$, it follows that $\|\widetilde{\bfW}\|_{\text{F}}^2 = \left\|V\bfA\right\|_{\text{F}}^2 = \left\|\bfA\right\|_{\text{F}}^2.$
Thus, \eqref{EQ:A_inequality} becomes
    \begin{align}
\left\|\bfW(t) -\bfW^\infty\right\|_{\text{F}}^2 &\leq \left\|\bfW(0) -\bfW^\infty\right\|_{\text{F}}^2e^{-\frac{2z_{\min}}{\varepsilon}t}. 
    \label{EQ:W_tilde_exp_bound}
    \end{align}
Since $\bfW = P^{\frac12} \bfU$, it follows that
    \begin{equation}
    \begin{split}
\left\|\bfW-\bfW^\infty\right\|_{\text{F}}^2 
  &= \sum_{k=1}^d\sum_{j=1}^N \left[\rho_j\left((\bfu_j)_k-\bfu^\infty_k\right)^2\right]
  \\
  &= \sum_{j=1}^N\rho_j\left\|\bfu_j-\bfu^\infty\right\|_2^2 \geq \min\{\rho_k\}\left\|\bfu_i-\bfu^\infty\right\|_2^2,
    \end{split}
    \end{equation}
for any $i\in\{1,\cdots,N\}$.
Thus, \eqref{EQ:W_tilde_exp_bound} gives the decay bound
    \begin{align}
\left\|\bfu_i(t)-\bfu^\infty\right\|_2^2
  &\leq \frac{1}{\min\{\rho_k\}}\left\|\bfW(t)-\bfW^\infty\right\|_{\text{F}}^2
  \leq C_{\bfU}^2e^{-\frac{2z_{\min}}{\varepsilon}t},
    \end{align}
where
$  C_{\bfU} = \frac{1}{\min\{\rho_k\}^\frac{1}{2}}\left\|\bfW^0-\bfW^\infty\right\|_{\text{F}}.  $
Thus, $\lim_{t\to\infty}\bfu_i(t)=\bfu^\infty,$ as desired.
    \end{proof}

    \subsection{Energy Relaxation Proof}
    \label{section:energyRelaxationProof}

    \begin{theorem}
    \label{theorem:energyRelaxation}
With the same assumptions as in \Cref{lemma:mixtureValuesAreConstant},   
$\lim_{t\to\infty}\bfE(t) = \bfE^\infty$ 
and 
    \begin{align}
    \label{eq:energy_decay}
\left\|\bfE(t)-\bfE^\infty\right\|_2 
&\leq C_1 e^{-\frac{\widehat{z}_{\min}}{\varepsilon}t}+C_2\frac{e^{-\frac{z_{\min}}{\varepsilon}t}-e^{-\frac{\widehat{z}_{\min}}{\varepsilon}t}}{\widehat{z}_{\min}-z_{\min}},
    \end{align}
where $C_1 > 0$ and $C_2>0$ are constants that depend on the initial conditions of the system \eqref{EQ:Full_ODE_System} and $\widehat{z}_{\min}$ is a lower bound on the positive eigenvalues of $\widehat{Z}$.
    \end{theorem}
    
    \begin{proof}
The proof is based on \eqref{EQ:Scaled_ODE_System}, which is equivalent to \eqref{EQ:Full_ODE_System}.
Set $\bfxi^\infty\coloneqq Q^{-\frac{1}{2}}\bfE^\infty$, and note that
    \begin{align}
\widehat{Z}\bfxi^\infty=\frac{|\bfu^\infty|^2}{2}Q^{-\frac{1}{2}}(F-B)M\mathbf{1}.
    \end{align}
Since $Q$ and $\bfE^\infty$ are time invariant, so too is $\bfxi^\infty$.
Hence $\widetilde{\bfxi}\coloneqq\bfxi-\bfxi^\infty$ satisfies the ODE
    \begin{align}
    \label{eq:xi_bar_odeSystem}
\frac{\diff\widetilde{\bfxi}}{\diff t} &= -\frac{1}{\varepsilon}\widehat{Z}\widetilde{\bfxi} + \frac{1}{2\varepsilon}Q^{-\frac{1}{2}}\left[(G-C)-|\bfu^\infty|^2(F-B)\right]M\mathbf{1}.
    \end{align}
Consider the null and range space components of $\widetilde{\bfxi}$ with respect to $\widehat{Z}$: $\widetilde{\bfxi}=\widetilde{\bfxi}_\mathcal{R}+\widetilde{\bfxi}_\mathcal{N}$.
It can be shown that $\widetilde{\bfxi}_\mathcal{N}\equiv\mathbf{0}$.(See \Cref{appendix:extraProofs}.) 
Therefore 
$ \widetilde{\bfxi}
    =\widetilde{\bfxi}_\mathcal{R}
    = \widehat{V} \bfb, $
where
$\widehat{V}=\left[\widehat{\bfv}_1,\cdots,\widehat{\bfv}_{N-1}\right]\in\R^{N\times(N-1)}$
is a matrix whose columns form an orthonormal basis for the range space of $\widehat{Z}$, and $\bfb = \widehat{V}^\top \widetilde{\bfxi} \in\R^{N-1}$.
Multiplication of \eqref{eq:xi_bar_odeSystem} by $\bfb^\top \widehat{V}^\top$ gives
    \begin{equation}
    \label{EQ:b-sqr-ODE}
\frac{1}{2}\frac{\diff}{\diff t}\left\|\bfb\right\|_2^2 
= -\frac{1}{\varepsilon}\bfb^\top \widehat{V}^\top \widehat{Z}\widehat{V}\bfb + \frac{1}{2\varepsilon}\bfb^\top \widehat{V}^\top Q^{-\frac{1}{2}}\left[(G-C)-|\bfu^\infty|^2(F-B)\right]M\mathbf{1}
.
    \end{equation}
The eigenvalues of $\widehat{V}^\top \widehat{Z}\widehat{V}$ are strictly positive and can be bounded below by
$  \widehat{z}_{\min} \coloneqq \frac{N\min_{i,j}B_{i,j}}{\max\{n_k\}} $
(see \Cref{appendix:eigenvalueBounds}),
which gives the bound $\bfb^\top \widehat{V}^\top \widehat{Z}\widehat{V}\bfb\geq\widehat{z}_{\min}\|\bfb\|_2^2$ for any $\bfb\in\R^{N-1}$.
With this bound and an application of the Cauchy-Schwarz inequality, \eqref{EQ:b-sqr-ODE} becomes
    \begin{equation}
\frac{\diff}{\diff t}\left\|\bfb\right\|_2^2
  \leq -\frac{2\widehat{z}_{\min}}{\varepsilon}\left\|\bfb\right\|_2^2
  + \frac{1}{\varepsilon}\frac{\|\bfb\|_2}{\min\{n_k\}^\frac{1}{2}}\left\|\left[(G-C)-|\bfu^\infty|^2(F-B)\right]M\mathbf{1}\right\|_2 .
    \label{EQ:normBode}
    \end{equation}
The next step is to bound the source term in \eqref{EQ:normBode}.
Standard norm inequalities and the definitions in \eqref{EQ:matrixDefinitions}, and \Cref{appendix:eigenvalueBounds}, give
    \begin{equation}
\left\|\left[\left(G-C\right)-|\bfu^\infty|^2\left(F-B\right)\right]M\mathbf{1}\right\|_2 
  \leq 2\max\{m_k\}\sum_{i=1}^N\sum\limits_{\substack{j=1 \\ j\neq i}}^NB_{\max}\left||\bfu_{i,j}|^2-|\bfu^\infty|^2\right|.
    \label{EQ:ineq1}
    \end{equation}
Meanwhile, to bound $\left||\bfu_{i,j}|^2-|\bfu^\infty|^2\right|$, we use (i) the decay bound in \eqref{eq:u_decay_bound}, (ii) the definitions of $\bfu_{i,j}$ and $\bfu^\infty$, (iii) the triangle inequality and Cauchy Schwarz inequality, (iv) the fact that $\alpha_{i,j}+\alpha_{j,i}=1$, and (v) the fact that $|\bfu_{i,j}| \leq \max_{i,j} \{ |\bfu_{i}|, |\bfu_{j}| \} \leq u_{\max}$ to obtain
    \begin{multline}
\left||\bfu_{i,j}|^2-|\bfu^\infty|^2\right|   
  = \left|(\bfu_{i,j}+\bfu^\infty)^\top (\bfu_{i,j}-\bfu^\infty) \right|
  \leq 2 u_{\max}\left\|\bfu_{i,j}-\bfu^\infty \right\|_2
    \\
  \leq 2 u_{\max} \left( \alpha_{i,j} \| \bfu_{i}-\bfu^\infty \|_2
  + \alpha_{j,i} \| \bfu_{j}-\bfu^\infty \|_2 \right) 
  \leq 2C_{\bfU}u_{\max}e^{-\frac{z_{\min}}{\varepsilon}t}.
    \label{EQ:ineq3}
    \end{multline}
Together \eqref{EQ:ineq1}
and \eqref{EQ:ineq3} give
    \begin{align}
\left\|\left[(G-C)-|\bfu^\infty|^2(F-B)\right]M\mathbf{1}\right\|_2 
&\leq 4N(N-1) C_{\bfU}u_{\max}B_{\max}\max\{m_k\}e^{-\frac{z_{\min}}{\varepsilon}t},
    \end{align}
and \eqref{EQ:normBode} becomes, after dividing through by $\left\|\bfb\right\|_2$,
    \begin{align}
\frac{\diff}{\diff t}\left\|\bfb\right\|_2 &\leq -\frac{\widehat{z}_{\min}}{\varepsilon}\left\|\bfb\right\|_2 +  \frac{1}{\varepsilon}C_0e^{-\frac{z_{\min}}{\varepsilon}t},
    \label{EQ:bODE_inequality}
    \end{align}
where $C_0 = \frac{1}{2}\frac{4N(N-1)C_{\bfU}u_{\max}B_{\max}\max\{m_k\}}{\min\{n_k\}^\frac{1}{2}}$, or equivalently, 
    \begin{align}
\frac{\diff}{\diff t}\left(\|\bfb\|_2e^{\frac{\widehat{z}_{\min}}{\varepsilon}t}\right) &\leq \frac{C_0}{\varepsilon}e^{\frac{\widehat{z}_{\min}-z_{\min}}{\varepsilon}t}.
    \end{align}
Integrating both sides above in $t$ gives
    \begin{align}
    \|\bfb\|_2 
     &\leq 
\|\bfb^0\|_2e^{-\frac{\widehat{z}_{\min}}{\varepsilon}t} + \frac{C_0}{\varepsilon}\frac{\varepsilon}{\widehat{z}_{\min}-z_{\min}}\left(e^{-\frac{z_{\min}}{\varepsilon}t}-e^{-\frac{\widehat{z}_{\min}}{\varepsilon}t}\right).
    \label{EQ:b_ineq_new}
    \end{align}
Since $\widehat{V}^\top \widehat{V}=I_{N-1}$, it follows that $\|\widetilde{\bfxi}\|_2^2 = \|\widehat{V}\bfb\|_2^2 = \left\|\bfb\right\|_2^2.$ 
Thus, \eqref{EQ:b_ineq_new} becomes
    \begin{align}
\left\|\bfxi-\bfxi^\infty\right\|_2 
\leq \left\|\bfxi^0-\bfxi^\infty\right\|_2e^{-\frac{\widehat{z}_{\min}}{\varepsilon}t} + C_0\frac{e^{-\frac{z_{\min}}{\varepsilon}t}-e^{-\frac{\widehat{z}_{\min}}{\varepsilon}t}}{\widehat{z}_{\min}-z_{\min}}.
    \end{align}
Further, since
    \begin{align}
\left\|\bfxi-\bfxi^\infty\right\|_2^2 &= \|Q^{-\frac{1}{2}}(\bfE-\bfE^\infty)\|_2^2 = \sum_{i=1}^Nn_i^{-1}|E_i-E^\infty_i|^2 \geq \frac{1}{\max\{n_k\}}\left\|\bfE-\bfE^\infty\right\|_2^2,
    \end{align}
    it follows that
    \begin{align}
\left\|\bfE(t)-\bfE^\infty\right\|_2 
&\leq \max\{n_k\}^\frac{1}{2}\left(\left\|\bfxi^0-\bfxi^\infty\right\|_2e^{-\frac{\widehat{z}_{\min}}{\varepsilon}t} + C_0 \frac{e^{-\frac{z_{\min}}{\varepsilon}t}-e^{-\frac{\widehat{z}_{\min}}{\varepsilon}t}}{\widehat{z}_{\min}-z_{\min}}\right),
    \label{EQ:EnergyDecayRates}
    \end{align}
which recovers \eqref{eq:energy_decay} with $C_1\coloneqq\max\{n_k\}^\frac{1}{2}\left\|\bfxi^0-\bfxi^\infty\right\|_2$ and $C_2\coloneqq C_0\max\{n_k\}^\frac{1}{2}$.
Thus $\lim_{t\to\infty}\bfE(t) = \bfE^\infty$, as desired.
    \end{proof}

    \subsection{Temperature Relaxation}
    \label{section:temperatureRelaxationProof}

    \begin{corollary}
    \label{corollary:temperatureRelaxation}
With the same assumptions as in \Cref{lemma:mixtureValuesAreConstant}, for each species $i\in\{1,\cdots,N\}$, $\lim_{t\to\infty}T_i(t)=T^\infty,$ and
    \begin{equation}
    \begin{split}
    \label{eq:temperature_decay_estimate}
\left|T_i-T^\infty\right| 
  &\leq \frac{2}{d\min\{n_k\}}\left[C_1 e^{-\frac{\widehat{z}_{\min}}{\varepsilon}t}+C_2\frac{e^{-\frac{z_{\min}}{\varepsilon}t}-e^{-\frac{\widehat{z}_{\min}}{\varepsilon}t}}{\widehat{z}_{\min}-z_{\min}}\right]
    \\
  & \qquad + \frac{\max\{m_k\}}{d}2u_{\max}C_{\bfU}e^{-\frac{z_{\min}}{\varepsilon}t}
  .
    \end{split}
    \end{equation}
    \end{corollary}
    
    \begin{proof}
 For each $i$,
    \begin{align}
T_i = \frac{2}{dn_i}E_i-\frac{m_i}{d}|\bfu_i|^2
    \qquad \text{and} \qquad
T^\infty &= \frac{2}{dn_i}E^\infty_i - \frac{m_i}{d}|\bfu^\infty|^2.
    \end{align}
Therefore 
    \begin{equation}
    \begin{split}
\left|T_i-T^\infty\right| 
  &\leq \frac{2}{d\min\{n_k\}}\left|E_i-E^\infty_i\right| + \frac{\max\{m_k\}}{d}\left||\bfu^\infty|^2-|\bfu_i|^2\right|. 
    \label{EQ:TempDiffs}         
    \end{split}
    \end{equation}
The velocity decay bound \eqref{eq:u_decay_bound},
gives
    \begin{equation}
    \begin{split}
\left||\bfu^\infty|^2-|\bfu_i|^2\right| 
  \leq 2 u_{\max} \|  \bfu^\infty - \bfu_i\|_2
  \leq 2u_{\max}C_{\bfU}e^{-\frac{z_{\min}}{\varepsilon}t}.     
    \end{split}
    \end{equation}
Using this and the energy decay bound \eqref{eq:energy_decay}, the result follows.
    \end{proof}
    
    \begin{remark}
The result in \Cref{corollary:temperatureRelaxation} can also be verified directly, using the formulation of the temperature equation \eqref{EQ:temperature_simple} and an approach similar to the proof of \Cref{theorem:energyRelaxation}.
The difference in the proofs comes from the fact that the nullspace component analogous to $\widetilde{\bfxi}_\mathcal{N}$ is no longer zero when using \eqref{EQ:temperature_simple}.
    \end{remark}

    \section{Numerical Demonstration}
    \label{section:numericalDemonstration}

In this section, we compute temperature and velocity profiles using a fully implicit (backward) Euler time stepping scheme.\footnote{The Backward Euler method preserves the monotonicity properties of the temperature and bulk velocity established in \Cref{section:existenceUniqueness}.
This fact will be will proved in a separate paper.}  
We assume a slab geometry in the $x_1$ direction.
Thus, while $d=3$, the quantities of interest depend only on $x_1$.
Moreover the  bulk velocities in the $x_2$ and $x_3$ direction are identically zero, i.e., $U_{i,j} = 0$ for all $j \in \{2,3\}$.

We consider systems with three different gas species chosen from a collection of four possible elements:  
Helium (He), Argon (Ar), Krypton (Kr), and Xenon (Xe).
The masses and diameters of these elements are taken from \cite{book:MathTheory_NonUniformGases_Chapman_Cowling} and are given below in SI units: 
    \begin{align}
  \begin{pmatrix}
    m_{\text{He}} \\ m_{\text{Ar}} \\ m_{\text{Kr}} \\ m_{\text{Xe}}
  \end{pmatrix}
=
  \begin{pmatrix}
    6.6464731 \\ 66.335209 \\ 139.14984 \\ 218.01714
  \end{pmatrix}
    \times 10^{-27} \text{ kg},
    \qquad
  \begin{pmatrix}
    d_{\text{He}} \\ d_{\text{Ar}} \\ d_{\text{Kr}} \\ d_{\text{Xe}}
  \end{pmatrix}
=
  \begin{pmatrix}
    2.193 \\ 3.659 \\ 4.199 \\ 4.939
  \end{pmatrix} \times 10^{-10} \text{ m}
    .
\end{align}
In all of the examples below, the monotonocity of the mininum temperature (Theorem \ref{thm:temperatureBoundedBelow}) and the upper and lower bounds on the bulk velocities (Theorem \ref{thm:velocityBounds}) are respected.

    \paragraph{Example 1: Temperature Decay, Ar-Kr-Xe Mixture}

The purpose of this example is to demonstrate temperature relaxation when the  bulk velocities are zero.  The initial number densities, velocities, and temperatures are given by
    \begin{align}
  \begin{pmatrix}
    n_{\text{Ar}}^0 \\ n_{\text{Kr}}^0 \\ n_{\text{Xe}}^0
  \end{pmatrix}
=
  \begin{pmatrix}
    1 \\ 1 \\ 1
  \end{pmatrix}
    \times 10^{28}\textnormal{ m}^{-3},
    \quad
\bfU^0
=
  \begin{pmatrix}
    0 & 0 & 0 \\
    0 & 0 & 0 \\
    0 & 0 & 0
  \end{pmatrix}
    \textnormal{ m/s}
    ,
    \quad
\bfT^0
=
  \begin{pmatrix}
    1000 \\ 2000 \\ 3000
  \end{pmatrix}
    \textnormal{ K}.
    \end{align}
\Cref{figure:tempDecay} shows the relaxation of the temperatures to their steady state value.
In this case the decay rate established in \eqref{eq:temperature_decay_estimate} is fairly sharp, even though the temperatures may not approach steady state monotonically.
\begin{figure}[ht]
    \centering
    \begin{subfigure}[t]{0.48\textwidth}
        \centering
        \includegraphics[width=\textwidth]{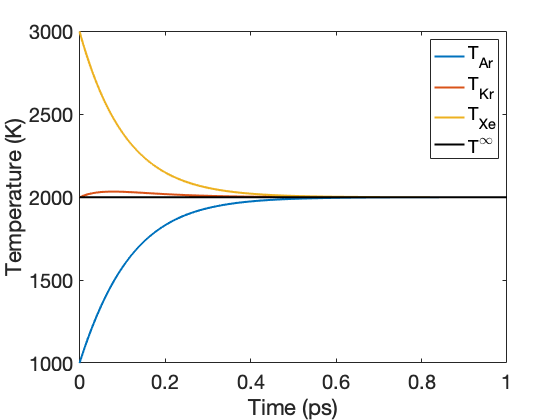}
        \caption{Species and steady state temperatures.}
    \end{subfigure}
    \hfill
    \begin{subfigure}[t]{0.48\textwidth}  
        \centering 
        \includegraphics[width=\textwidth]{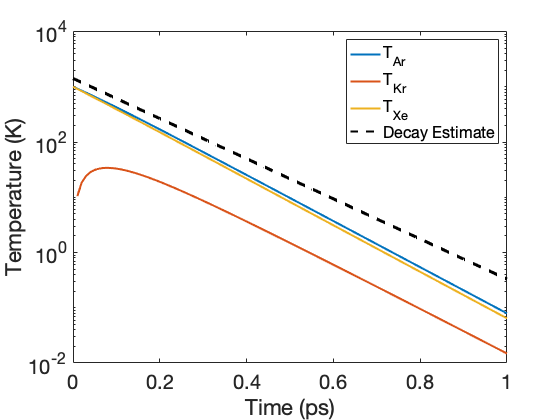}
        \caption{Deviation from steady state temperature and analytical decay estimate from \eqref{eq:temperature_decay_estimate}.}
    \end{subfigure}
    \caption{Temperature Decay of Ar-Kr-Xe Mixture in Example 1.}
    \label{figure:tempDecay}
\end{figure}

    \paragraph{Example 2: Velocity Decay,  Ar-Kr-Xe Mixture}
    
The purpose of this example is to demonstrate the velocity relaxation, with minimal effects from the temperature.
A positive velocity is given to one particle type, and each temperature is set to the same constant. 
The initial number densities, velocities, and temperatures are given by
\begin{align}
  \begin{pmatrix}
    n_{\text{Ar}}^0 \\ n_{\text{Kr}}^0 \\ n_{\text{Xe}}^0
  \end{pmatrix}
=
  \begin{pmatrix}
    3 \\ 2 \\ 1
  \end{pmatrix}
    \times 10^{28}\textnormal{ m}^{-3}
    ,
    \;\;\;
\bfU^0
=
  \begin{pmatrix}
    100 & 0 & 0 \\
    0 & 0 & 0 \\
    0 & 0 & 0
  \end{pmatrix}
    \textnormal{ m/s}
    ,
    \;\;\;
\bfT^0
=
  \begin{pmatrix}
    1000 \\ 1000 \\ 1000
  \end{pmatrix}
    \textnormal{ K}.
    \end{align}
\Cref{figure:VelDecay100} shows the relaxation of the velocities and temperatures to the steady state values.
The decay rates established for the velocities and temperatures in \eqref{eq:u_decay_bound} and \eqref{eq:temperature_decay_estimate} underestimate the true decay rates in the example.
The results also demonstrate that, unlike the two-species case, particle velocities may not converge to the steady-state value monotonically. 
\begin{figure}[ht]
    \centering
    \begin{subfigure}[t]{0.48\textwidth}
        \centering \includegraphics[width=\textwidth]{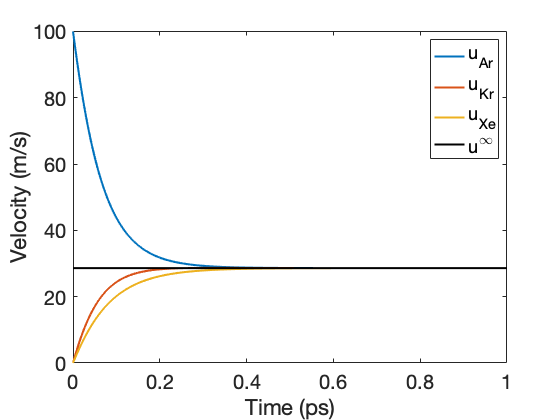}
        \caption{Species and steady state velocities.}
    \end{subfigure}
    \hfill
    \begin{subfigure}[t]{0.48\textwidth}  
        \centering         \includegraphics[width=\textwidth]{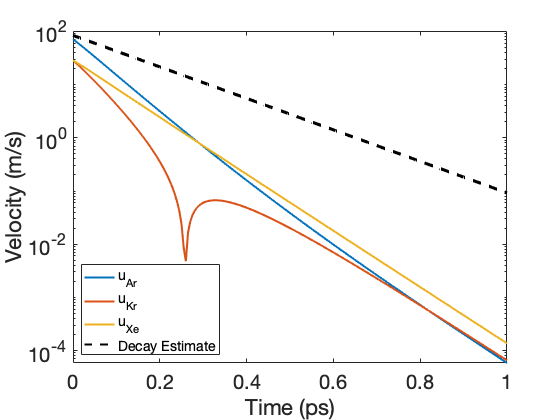}
        \caption{Deviation from steady state velocity and analytical decay estimate from \eqref{eq:u_decay_bound}.}
    \end{subfigure}
        \vskip\parskip
    \begin{subfigure}[t]{0.48\textwidth}   
        \centering 
        \includegraphics[width=\textwidth]{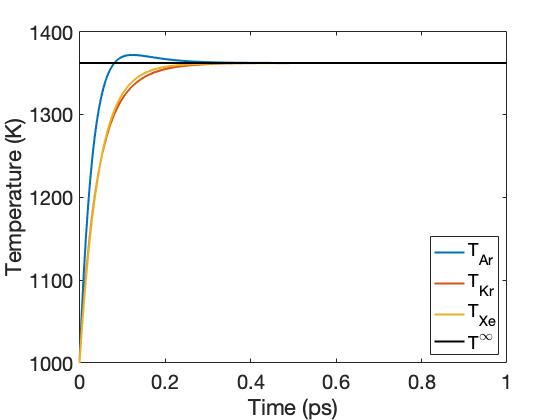}
        \caption{Species and steady state temperatures.}
    \end{subfigure}
    \hfill
    \begin{subfigure}[t]{0.48\textwidth}   
        \centering
        \includegraphics[width=\textwidth]{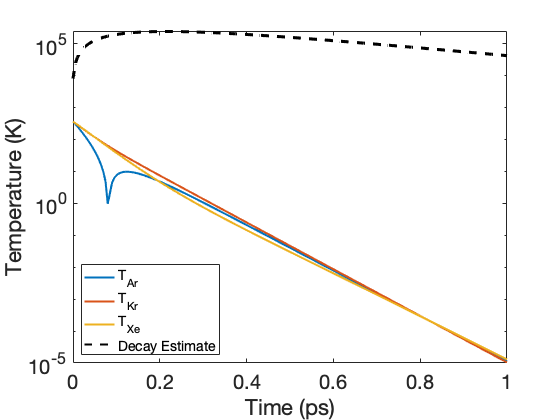}
        \caption{Deviation from steady state velocity and analytical decay estimate from \eqref{eq:temperature_decay_estimate}.}
    \end{subfigure}
    \caption{Velocity and Temperature Decay of Ar-Kr-Xe Mixture in Example 2.}
    \label{figure:VelDecay100}
\end{figure}

    \paragraph{Example 3: Velocity-Temperature Relaxation, He-Kr-Xe Mixture}
    
This test case exercises the model due to the large kinetic energy differences between the species.
The initial number densities, velocities, and temperatures are given by
    \begin{align}
  \begin{pmatrix}
    n_{\text{He}}^0 \\ n_{\text{Kr}}^0 \\ n_{\text{Xe}}^0
  \end{pmatrix}
=
  \begin{pmatrix}
    0.01 \\ 1 \\ 1
  \end{pmatrix}
    \times 10^{28} \textnormal{ m}^{-3}
    ,\quad
\bfU^{0}
=
  \begin{pmatrix}
    864.8 & 0 & 0 \\
    0 & 0 & 0 \\
    0 & 0 & 0
  \end{pmatrix}
    \textnormal{ m/s}
    ,\quad
\bfT^0
=
  \begin{pmatrix}
    3000 \\ 300 \\ 300
  \end{pmatrix}
    \textnormal{ K}
    .
    \end{align}
The momentum equilibration has a corresponding kinetic energy redistribution
which causes non-monotonic changes in the temperature profile, as seen in 
\Cref{figure:JeffProblem_temperature}.
Because of the large differences in the particle masses and number densities, the estimates for the decay rates established in \eqref{eq:u_decay_bound} and \eqref{eq:temperature_decay_estimate} are very weak.
\begin{figure}[ht]
    \centering
    \begin{subfigure}[t]{0.48\textwidth}
        \centering
        \includegraphics[width=\textwidth]{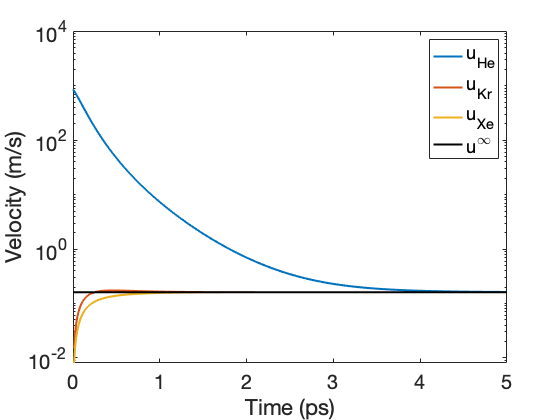}
        \caption{Species and steady state velocities.}
    \end{subfigure}
    \hfill
    \begin{subfigure}[t]{0.48\textwidth}  
        \centering 
        \includegraphics[width=\textwidth]{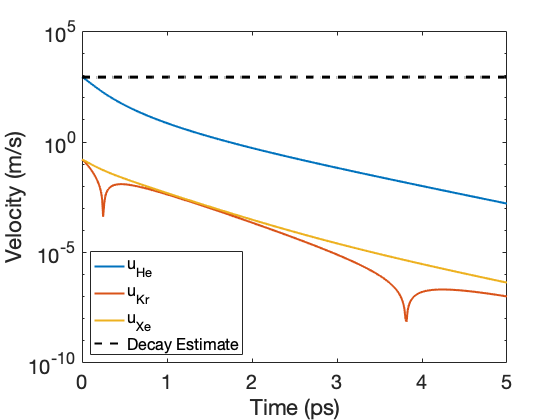}
        \caption{Deviation from steady state velocity and analytical decay estimate from \eqref{eq:u_decay_bound}.}
    \end{subfigure}
        \vskip\parskip
    \begin{subfigure}[t]{0.48\textwidth}   
        \centering 
        \includegraphics[width=\textwidth]{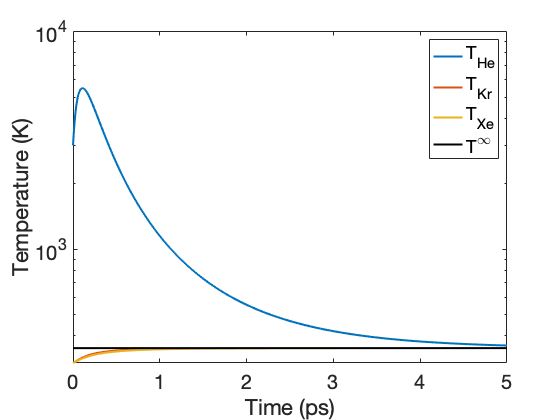}
        \caption{Species and steady state temperatures.}    
        \label{figure:JeffProblem_temperature}
    \end{subfigure}
    \hfill
    \begin{subfigure}[t]{0.48\textwidth}   
        \centering 
        \includegraphics[width=\textwidth]{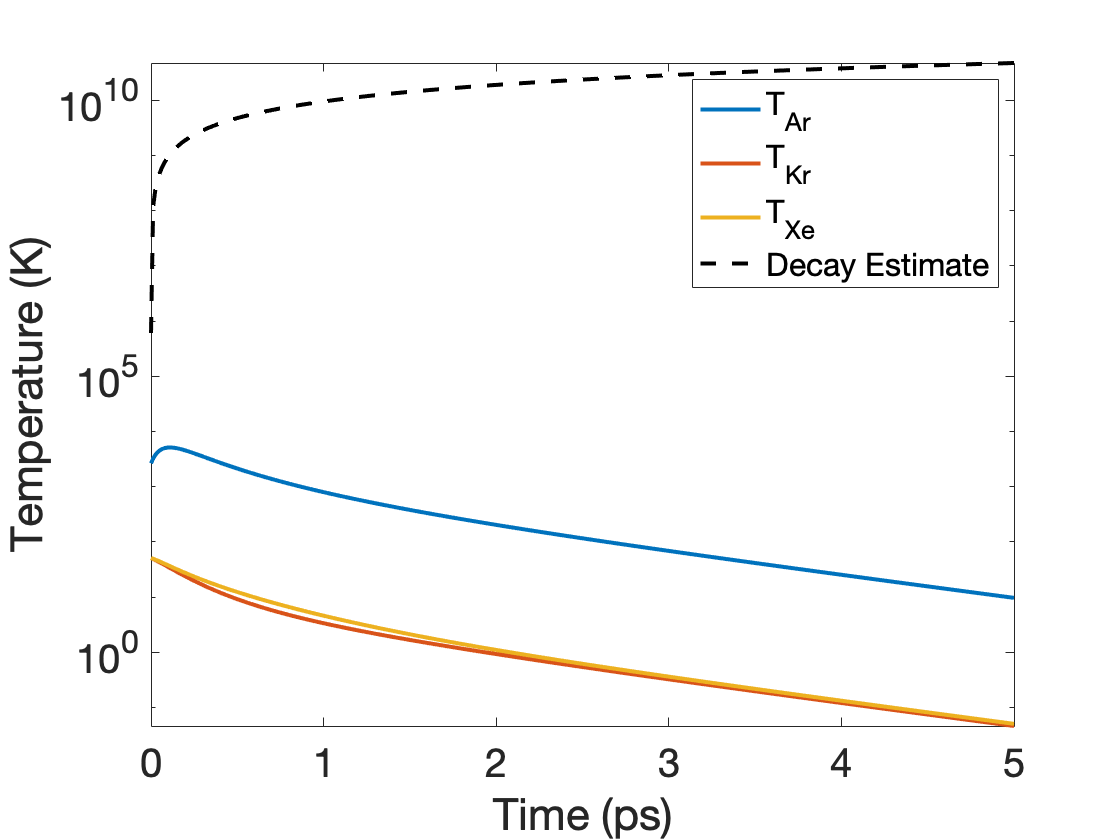}
        \caption{Deviation from steady state velocity and analytical decay estimate from \eqref{eq:temperature_decay_estimate}.}
    \end{subfigure}
    \caption{Velocity and Temperature Decay of He-Kr-Xe Mixture in Example 3.}
    \label{figure:JeffProblem}
\end{figure}
    
    \section{Conclusions}
    \label{sec:conclusions}

In this paper, we have studied the moment equations associated to a recently developed multi-species BGK kinetic model for rarefied gas dynamics in the spatially homogeneous setting.
The model includes collision frequencies that, unlike most models in the literature, depend on the species temperatures.
This fact complicates the analysis of the moment equations and subsequent numerical tools.
We have proven that all species temperatures are bounded below by a positive, non-decreasing  temperature envelope, thus establishing that species temperatures always remain positive.
We have also established upper and lower bounds on the components of the bulk velocity that are in terms of the initial data.

Using the lower bound on the species temperatures, we have shown that the moments always stay within a bounded, time-invariant set of physically realizable states, which, in turn, leads to the existence of global, unique solutions.
Finally, we have proved that unique equilibria exist for the moment equations, and solutions converge to these equilbria exponentially in time.
In addition, we have established bounds on the convergence rates to equilibria.
We concluded the paper with some basic numerical simulations that demonstrate some of the established theoretical behavior.

The numerical results in this paper are computed with a backward Euler method, using an iterative scheme to solve the relevant non-linear algebraic equations at each time-step.
In future work, we will describe the scheme more fully and prove convergence under a suitable time step restriction that does not require resolution of the parameter $\varepsilon$.
We will use the backward Euler method as  a component in an IMEX scheme for simulating multi-species BGK models that include phase-space advection.

    \appendix

    \section{Picard-Lindel\"{o}f Theorem}

We will make use of the following version of the Picard-Lindel\"{o}f Theorem for autonomous systems:

    \begin{theorem}
    \label{thm:picard-lindelof}
Suppose that $\bfu_0\in\R^d$ and $\bff\in C(\overline{B}(\bfu_0,\beta);\R^d)$, where $B(\bfu_0,\beta)$ ($\overline{B}(\bfu_0,\beta)$) denotes the open (closed) Euclidean ball of radius $\beta$ centered at $\bfu_0$.
Suppose further that $\|\bff(\bfv)\|_2 \le M$, for all $\bfv\in \overline{B}(\bfu_0,\beta)$, and $\bff$ is Lipschitz continuous, with constant $L>0$, on $\overline{B}(\bfu_0,\beta)$.
Set $\delta = \min\left(\frac{1}{2L},\frac{\beta}{M} \right)$.
Then the IVP $\bfu'(t) = \bff(\bfu(t))$, $\bfu(t_0) = \bfu_0$ has a unique solution on the interval $[t_0-\delta,t_0+\delta]$.
    \end{theorem}

    \section{Bounding the Eigenvalues of \texorpdfstring{$Z$}{Z} and \texorpdfstring{$\widehat{Z}$}{Zhat}}
    \label{appendix:eigenvalueBounds}
    
The matrices $Z$ and $\widehat{Z}$, defined in \eqref{eq:Z_Zhat}, play a key role in the dynamics of \eqref{EQ:Full_ODE_System} and the equivalent system \eqref{EQ:Scaled_ODE_System}.
In this section, we prove bounds on the eigenvalues of these matrices.
With the collision frequencies defined in \eqref{eqn:collision-freqs-hhm}, the matrices $A$ and $B$, defined in \eqref{eq:ABSC_defs}, take the form
    \begin{align}
A_{i,j} &= \frac{16}{3}\sqrt{\frac{\pi}{2}}\frac{m_im_j(d_i+d_j)^2}{(m_i+m_j)^3}\rho_i\rho_j\sqrt{\frac{T_i}{m_i}+\frac{T_j}{m_j}},
    \\
B_{i,j} &= \frac{8}{3}\sqrt{\frac{\pi}{2}}\frac{(d_i+d_j)^2}{(m_i+m_j)^2}\rho_i\rho_j\sqrt{\frac{T_i}{m_i}+\frac{T_j}{m_j}}.
    \end{align}
Since each value used to define the matrix elements is positive and bounded above and below (importantly, the temperature), then their values can be bounded as follows:
    \begin{align}
A_{\min}(t)
    \coloneqq
\min_{i,j}A_{i,j}(t)
    \leq
A_{i,j}(t)
    \leq
\max_{i,j} A_{i,j}(t)
    \eqqcolon
A_{\max}(t),
    \\
B_{\min}(t)
    \coloneqq
\min_{i,j}B_{i,j}(t)
    \leq
B_{i,j}(t)
    \leq
\max_{i,j}B_{i,j}(t)
    \eqqcolon
B_{\max}(t).
    \end{align}

    \begin{theorem}[Eigenvalue Bounds]
Let $A,B,D,F \in \R^{N \times N}$ be the matrices defined in \eqref{EQ:matrixDefinitions}, and let 
$A_{\min} = \min_{i,j} A_{i,j}$, $B_{\min} = \min_{i,j} B_{i,j}$, $A_{\max} = \max_{i,j} A_{i,j}$, and $B_{\max} = \max_{i,j} B_{i,j}$.
Then the eigenvalues of the matrices $Z=P^{-\frac{1}{2}}(D-A)P^{-\frac{1}{2}}$ and $\widehat{Z}=Q^{-\frac{1}{2}}(F-B)Q^{-\frac{1}{2}}$ satisfy
    \begin{align}
0 &=z_0<z_{\min}\leq z_1 \leq \cdots \leq z_{N-1} \leq z_{\max},
    \label{EQ:inequality1}
    \\
0 &=\widehat{z}_0<\widehat{z}_{\min}\leq \widehat{z}_1 \leq \cdots \leq \widehat{z}_{N-1} \leq \widehat{z}_{\max},
    \label{EQ:inequality2}
    \end{align}
where
    \begin{align}
z_{\min} = \frac{A_{\min}N}{\max\{\rho_k\}}
    &,\quad
z_{\max} = \frac{A_{\max}(N-1)}{\min\{\rho_k\}},\quad 
    \\
\widehat{z}_{\min} = \frac{B_{\min}N}{\max\{n_k\}}
    &,\quad
\widehat{z}_{\max} = \frac{B_{\max}(N-1)}{\min\{n_k\}}.
    \end{align}
    \end{theorem}
    
    \begin{proof}
The proofs of \eqref{EQ:inequality1} and \eqref{EQ:inequality2} are similar; thus we prove \eqref{EQ:inequality1} and leave the remaining details to the reader.    
The (single) zero eigenvalue of $(D-A)$ corresponds to the eigenvector, $\mathbf{1}$, spanning the null space of $(D-A)$.
To find bounds on the other eigenvalues, let $\bfy\in\mathcal{R}(D-A)=\mathcal{N}(D-A)^\perp$, so that $\bfy^\top\mathbf{1}=\sum_iy_i=0$.
Recall that $\bfy^\top(D-A)\bfy = \frac{1}{2}\sum_{i,j}A_{i,j}(y_i-y_j)^2$.
Using $\sum_iy_i=0$,
    \begin{align}
\bfy^\top(D-A)\bfy
  \leq \frac{A_{\max}}{2}\sum_{i,j}(y_i-y_j)^2
  = \frac{A_{\max}}{2}\sum_{i}\sum_{j\neq i}\left(y_i^2+y_j^2\right) 
  = {A_{\max}}(N-1)\left\|\bfy\right\|_2^2.
    \end{align}
Thus if $\bfy=P^{-\frac{1}{2}}\bfz,$ then
    \begin{multline}
\bfz^\top Z\bfz
  \leq (N-1)A_{\max}\bfz^\top P^{-1}\bfz
  \\
  = (N-1)A_{\max}\sum_i\frac{1}{\rho_i}z_i^2 
  \leq \frac{(N-1)A_{\max}}{\min\{\rho_k\}}\left\|\bfz\right\|_2^2
  \eqqcolon z_{\max}\left\|\bfz\right\|_2^2,         
    \end{multline}
which gives the upper bound in \eqref{EQ:inequality1}.
For the lower bound,
    \begin{align}
\bfy^\top(D-A)\bfy
  \geq \frac{A_{\min}}{2}\sum_{i=1}^N\sum_{j=1}^N(y_i^2+y_j^2) 
  = A_{\min}N\left\|\bfy\right\|_2^2.
    \end{align}
Thus if $\bfy=P^{-\frac{1}{2}}\bfz,$ then
    \begin{align}
\bfz^\top Z\bfz
  &\geq NA_{\min}\bfz^\top P^{-1}\bfz
  \geq \frac{NA_{\min}}{\max\{\rho_k\}}\left\|\bfz\right\|_2^2
  \eqqcolon z_{\min}\left\|\bfz\right\|_2^2.
    \end{align}
    \end{proof}
    
    \section{Null Space Components of \texorpdfstring{$\widetilde{\bfW}$}{WTilde} and \texorpdfstring{$\widetilde{\bfxi}$}{XiTilde}}
    \label{appendix:extraProofs}
    
In this section, we show that $\widetilde{\bfW}_\mathcal{N}$ (the projection of $\widetilde{\bfW}$ onto $\mathcal{N}(Z)$) and $\widetilde{\bfxi}_\mathcal{N}$ (the projection of $\widetilde{\bfxi}$ onto $\mathcal{N}(\widehat{Z})$) are both zero.
These results are used in the proofs of \Cref{theorem:velocityRelaxationProof,theorem:energyRelaxation}, respectively.

    \begin{lemma}
$\widetilde{\bfW}_\mathcal{N}\equiv\mathbf{0}$.
    \end{lemma}
    
    \begin{proof}
Since the null space of $Z$ is spanned by $P^\frac{1}{2}\mathbf{1}$, it is sufficient to show that $\widetilde{\bfW}^\top P^\frac{1}{2}\mathbf{1}=\mathbf{0}$.
Write $\bfu^\infty=\frac{\bfU^\top P\mathbf{1}}{\mathbf{1}^\top P\mathbf{1}}$ and $\left(\bfU^\infty\right)^\top=\bfu^\infty\mathbf{1}^\top$, so that 
    \begin{align}
\widetilde{\bfW}^\top P^\frac{1}{2}\mathbf{1}
  = \bfU^\top P\mathbf{1}-\left(\bfU^\infty\right)^\top P\mathbf{1}
  = \bfU^\top P\mathbf{1}-\frac{\bfU^\top P\mathbf{1}}{\mathbf{1}^\top P\mathbf{1}}\mathbf{1}^\top P\mathbf{1}
  = \mathbf{0}.
    \end{align}
    \end{proof}

    \begin{lemma}
$\widetilde{\bfxi}_\mathcal{N} \equiv \mathbf{0}$.
    \end{lemma}
    
    \begin{proof}
Since the null space of $\widehat{Z}$ is spanned by $Q^\frac{1}{2}\mathbf{1}$, it is sufficient to show that $\widetilde{\bfxi}^\top Q^\frac{1}{2}\mathbf{1}=0.$
Recall that $\widetilde{\bfxi}=Q^{-\frac{1}{2}}\left(\bfE-\bfE^\infty\right);$
hence
    \begin{align}
\widetilde{\bfxi}^\top Q^\frac{1}{2}\mathbf{1}
  = \left(\bfE-\bfE^\infty\right)^\top\mathbf{1}
  = \sum_{i=1}^NE_i - \sum_{i=1}^NE^\infty_i
  = 0,
    \end{align}
where the fact that $\sum_iE_i=\sum_iE_i^\infty$ is a consequence of \eqref{EQ:TMixIsConstant} in the proof of \Cref{lemma:mixtureValuesAreConstant}.
    \end{proof}

    \bibliographystyle{plain}
    \bibliography{references}

    \end{document}